\documentclass[pra,superscriptaddress,nopacs,twocolumn]{revtex4-2}

\usepackage[utf8x]{inputenc}

\usepackage{enumerate}
\usepackage{comment}
\usepackage[qm]{qcircuit}
\usepackage{xcolor}
\definecolor{darkred}{rgb}{0.8,0.1,0.1}
\definecolor{darkblue2}{rgb}{0.1,0.2,0.8}
\usepackage{amsmath,mathtools}
\usepackage{enumerate}
\usepackage{graphicx,epsfig,amsmath,amssymb,verbatim,color}
\usepackage{dsfont}
\usepackage{float}
\usepackage{tikz}
\usepackage{hyperref}

\newcommand{\diff}[1]{\textcolor{black}{ #1}}

\newcommand{\lb}[1]{t(#1)}
\usepackage{enumitem}
\usepackage{hyperref}
\hypersetup{colorlinks=true,citecolor=blue,linkcolor=blue,filecolor=blue,urlcolor=blue,breaklinks=true}
\usepackage{cleveref}
\usepackage{graphicx}
\usepackage{xcolor}

\definecolor{darkblue}{RGB}{0,76,156}
\definecolor{darkkblue}{RGB}{0,0,153}
\definecolor{blue2}{RGB}{102,178,255}

\usepackage{url}
\newtheorem{definition}{Definition}
\newtheorem{proposition}{Proposition}
\newtheorem{lemma}[proposition]{Lemma}

\newtheorem{theorem}[proposition]{Theorem}

\newtheorem{corollary}[proposition]{Corollary}

\def\squareforqed{\hbox{\rlap{$\sqcap$}$\sqcup$}}
\def\qed{\ifmmode\squareforqed\else{\unskip\nobreak\hfil
\penalty50\hskip1em\null\nobreak\hfil\squareforqed
\parfillskip=0pt\finalhyphendemerits=0\endgraf}\fi}
\def\endenv{\ifmmode\;\else{\unskip\nobreak\hfil
\penalty50\hskip1em\null\nobreak\hfil\;
\parfillskip=0pt\finalhyphendemerits=0\endgraf}\fi}
\newenvironment{proof}{\noindent \textbf{{Proof.}~}}{\hfill $\blacksquare$}

\mathchardef\ordinarycolon\mathcode`\:
\mathcode`\:=\string"8000
\def\vcentcolon{\mathrel{\mathop\ordinarycolon}}
\begingroup \catcode`\:=\active
  \lowercase{\endgroup
  \let :\vcentcolon
  }

\makeatletter
\def\resetMathstrut@{%
    \setbox\z@\hbox{%
        \mathchardef\@tempa\mathcode`\[\relax
        \def\@tempb##1"##2##3{\the\textfont"##3\char"}%
        \expandafter\@tempb\meaning\@tempa \relax
    }%
    \ht\Mathstrutbox@\ht\z@ \dp\Mathstrutbox@\dp\z@}
\newcommand{\nc}{\newcommand}
\nc{\rnc}{\renewcommand}
\nc{\lbar}[1]{\overline{#1}}
\nc{\bra}[1]{\langle#1|}
\nc{\ket}[1]{|#1\rangle}
\nc{\ketbra}[2]{|#1\rangle\!\langle#2|}
\nc{\braket}[2]{\langle#1|#2\rangle}

\nc{\proj}[1]{| #1\rangle\!\langle #1 |}
\nc{\avg}[1]{\langle#1\rangle}
\nc{\Rank}{\operatorname{Rank}}
\nc{\smfrac}[2]{\mbox{$\frac{#1}{#2}$}}
\nc{\tr}{\operatorname{Tr}}
\nc{\ox}{\otimes}
\nc{\dg}{\dagger}
\nc{\dn}{\downarrow}
\nc{\cA}{{\cal A}}
\nc{\cB}{{\cal B}}
\nc{\cC}{{\cal C}}
\nc{\cD}{{\cal D}}
\nc{\cE}{{\cal E}}
\nc{\cF}{{\cal F}}
\nc{\cG}{{\cal G}}
\nc{\cH}{{\cal H}}
\nc{\cI}{{\cal I}}
\nc{\cJ}{{\cal J}}
\nc{\cK}{{\cal K}}
\nc{\cL}{{\cal L}}
\nc{\cM}{{\cal M}}
\nc{\cN}{{\cal N}}
\nc{\cO}{{\cal O}}
\nc{\cP}{{\cal P}}
\nc{\cQ}{{\cal Q}}
\nc{\cR}{{\cal R}}
\nc{\cS}{{\cal S}}
\nc{\cT}{{\cal T}}
\nc{\cV}{{\cal V}}
\nc{\cX}{{\cal X}}
\nc{\cY}{{\cal Y}}
\nc{\cZ}{{\cal Z}}
\nc{\cW}{{\cal W}}
\nc{\csupp}{{\operatorname{csupp}}}
\nc{\qsupp}{{\operatorname{qsupp}}}
\nc{\var}{{\operatorname{var}}}
\nc{\rar}{\rightarrow}
\nc{\lrar}{\longrightarrow}
\nc{\polylog}{{\operatorname{polylog}}}
\nc{\wt}{{\operatorname{wt}}}
\nc{\dia}{{\diamondsuit }}

\def\x{\xi}

\nc{\SSS}{{{\mathbb S}}}
\nc{\RR}{{{\mathbb R}}}
\nc{\CC}{{{\mathbb C}}}
\nc{\FF}{{{\mathbb F}}}
\nc{\NN}{{{\mathbb N}}}
\nc{\ZZ}{{{\mathbb Z}}}
\nc{\PP}{{{\mathbb P}}}
\nc{\QQ}{{{\mathbb Q}}}
\nc{\UU}{{{\mathbb U}}}
\nc{\EE}{{{\mathbb E}}}
\nc{\id}{{\operatorname{id}}}

\nc{\CHSH}{{\operatorname{CHSH}}}

\nc{\<}{\langle}
\rnc{\>}{\rangle}
\nc{\Hom}[2]{\mbox{Hom}(\CC^{#1},\CC^{#2})}
\nc{\rU}{\mbox{U}}

\nc{\ob}[1]{#1}

\nc{\SEP}{{\text{SEP}}}
\nc{\NS}{{\text{NS}}}
\nc{\LOCC}{{\operatorname{LOCC}}}
\nc{\PPT}{{\operatorname{PPT}}}
\nc{\EXT}{{\text{EXT}}}
\nc{\Sym}{{\operatorname{Sym}}}

\nc{\HH}{\mathbb{H}}

\nc{\ERLO}{{E_{\text{r,LO}}}}
\nc{\ERLOCC}{{E_{\text{r,LOCC}}}}
\nc{\ERPPT}{{E_{\text{r,PPT}}}}
\nc{\ERLOCCinfty}{{E^{\infty}_{\text{r,LOCC}}}}
\nc{\Aram}{{\operatorname{\sf A}}}

\let\id\1

\usepackage{tikz}
\usetikzlibrary{arrows}

\makeatletter
\def\grd@save@target#1{%
  \def\grd@target{#1}}
\def\grd@save@start#1{%
  \def\grd@start{#1}}
\tikzset{
  grid with coordinates/.style={
    to path={%
      \pgfextra{%
        \edef\grd@@target{(\tikztotarget)}%
        \tikz@scan@one@point\grd@save@target\grd@@target\relax
        \edef\grd@@start{(\tikztostart)}%
        \tikz@scan@one@point\grd@save@start\grd@@start\relax
        \draw[minor help lines,magenta] (\tikztostart) grid (\tikztotarget);
        \draw[major help lines] (\tikztostart) grid (\tikztotarget);
        \grd@start
        \pgfmathsetmacro{\grd@xa}{\the\pgf@x/1cm}
        \pgfmathsetmacro{\grd@ya}{\the\pgf@y/1cm}
        \grd@target
        \pgfmathsetmacro{\grd@xb}{\the\pgf@x/1cm}
        \pgfmathsetmacro{\grd@yb}{\the\pgf@y/1cm}
        \pgfmathsetmacro{\grd@xc}{\grd@xa + \pgfkeysvalueof{/tikz/grid with coordinates/major step}}
        \pgfmathsetmacro{\grd@yc}{\grd@ya + \pgfkeysvalueof{/tikz/grid with coordinates/major step}}
        \foreach \x in {\grd@xa,\grd@xc,...,\grd@xb}
        \node[anchor=north] at (\x,\grd@ya) {\pgfmathprintnumber{\x}};
        \foreach \y in {\grd@ya,\grd@yc,...,\grd@yb}
        \node[anchor=east] at (\grd@xa,\y) {\pgfmathprintnumber{\y}};
      }
    }
  },
  minor help lines/.style={
    help lines,
    step=\pgfkeysvalueof{/tikz/grid with coordinates/minor step}
  },
  major help lines/.style={
    help lines,
    line width=\pgfkeysvalueof{/tikz/grid with coordinates/major line width},
    step=\pgfkeysvalueof{/tikz/grid with coordinates/major step}
  },
  grid with coordinates/.cd,
  minor step/.initial=.2,
  major step/.initial=1,
  major line width/.initial=2pt,
}
\makeatother

\tikzset{
  treenode/.style = {align=center, inner sep=0pt, text centered,
    font=\sffamily},
  arn_n/.style = {treenode, circle, white, font=\sffamily\bfseries, draw=black,
    fill=black, text width=1.5em},% arbre rouge noir, noeud noir
  arn_r/.style = {treenode, circle, red, draw=red, 
    text width=1.5em, very thick},% arbre rouge noir, noeud rouge
  arn_x/.style = {treenode, rectangle, draw=black,
    minimum width=0.5em, minimum height=0.5em}% arbre rouge noir, nil
}

\usepackage{pdfpages}
\makeatletter
\AtBeginDocument{\let\LS@rot\@undefined}
\makeatother
\nc{\Pauli}{{{\operatorname{Pauli}}}}
\nc{\Spec}{{{\operatorname{Spec}}}}
\nc{\Stab}{{{\operatorname{Stab}}}}
\nc{\PWP}{{{\operatorname{CPWP}}}}
 \nc{\SCPO}{{{\operatorname{CSPO}}}}
 \nc{\bu}{{{\textbf{u}}}}
 \nc{\bx}{{{\textbf{x}}}}
  \nc{\bj}{{{\textbf{j}}}}
  \nc{\bz}{{{\textbf{z}}}}
  \nc{\bv}{{{\textbf{v}}}}
    \nc{\bw}{{{\textbf{w}}}}
    \nc{\bs}{{{\textbf{s}}}}
        \nc{\bb}{{{\textbf{b}}}}
\nc{\bc}{{{\textbf{c}}}}
 \nc{\bt}{{{\textbf{t}}}}
  \nc{\bp}{{{\textbf{p}}}}
    \nc{\bq}{{{\textbf{q}}}}
        \nc{\by}{{{\textbf{y}}}}
  \nc{\sn}{{{\operatorname{sn}}}}
\allowdisplaybreaks

\begin{document}
\title{  Lower bound for the $T$ count via unitary stabilizer nullity}
\author{Jiaqing Jiang}
\email{jiaqingjiang95@gmail.com}
\affiliation{Institute for Quantum Computing, Baidu Research, Beijing 100193, China}
\affiliation{Computing and Mathematical Sciences, California Institute of Technology, Pasadena, California, USA}

\author{Xin Wang}
\email{wangxin73@baidu.com}
\affiliation{Institute for Quantum Computing, Baidu Research, Beijing 100193, China}

\begin{abstract}
We introduce magic measures to quantify the nonstabilizerness of multiqubit quantum gates and establish lower bounds on the $T$ count for fault-tolerant quantum computation. First, we introduce the stabilizer nullity of multi-qubit unitary, which is based on the subgroup of the quotient Pauli group associated with the unitary. This unitary stabilizer nullity extends the state-stabilizer nullity by Beverland \emph{et al.} to a dynamic version. In particular, we show this nonstabilizerness measure has desirable properties such as subadditivity under composition and additivity under tensor product. Second, we prove that a given unitary's stabilizer nullity is a lower bound for the $T$ count, utilizing the above properties in gate synthesis. Third, we compare the state- and the unitary-stabilizer nullity, proving that the lower bounds for the $T$ count obtained by the unitary-stabilizer nullity are never less than the state-stabilizer nullity. Moreover, we show an explicit $n$-qubit unitary family of unitary-stabilizer nullity $2n$, which implies that its $T$ count is at least $2n$. This gives an example where the bounds derived by the unitary-stabilizer nullity  strictly outperform the state-stabilizer nullity by a factor of $2$. We finally showcase the advantages of  unitary-stabilizer nullity in estimating the $T$ count of quantum gates with interests.
\end{abstract}  
\date{\today}
\maketitle 
% \tableofcontents

\section{Introduction}
A quantum circuit comprised of only Clifford gates confers no computational advantage since it can be efficiently simulated on a classical computer~\cite{gottesman1998heisenberg,aaronson2004improved}.
However, the addition of $T$ gates into Clifford circuits~\cite{dawson2006solovay,nielsen2002quantum} can achieve universal quantum computation.
The minimum number of $T$ gates used, or the \textit{$T$ count}, plays a role in quantum computation. \diff{The study of $T$ count has applications in quantum simulation, quantum error correction, and quantum circuit synthesis. Specifically, in} the theoretical study of classically simulating quantum computation, $T$ count serves as a quantifier of difficulty for such simulation in many algorithms~\cite{bravyi2016improved,bravyi2016trading,bravyi2019simulation}. It is shown that quantum circuits composed of Clifford gates and $t$ $T$ gates can be simulated on a classical computer in time, which is exponential in $t$~\cite{bravyi2016improved,bravyi2019simulation, seddon2021quantifying, kocia2020improved}.
\diff{Better estimation of  the $T$-count gives better analysis of the performance for those algorithms.}
On the other hand, in the framework of fault-tolerant quantum computation~\cite{zhou2000methodology,Bravyi_2005},  $T$ count dominates the cost of computation. This is because in many promising fault-tolerant architectures, the Clifford gates can be implemented transversally~\cite{gottesman1997stabilizer} and the required $T$ gate for universality can be performed via state injection~\cite{gottesman1999demonstrating,zhou2000methodology}. 
The key of this resolution is to perform
magic state distillation~\cite{Bravyi_2005,bravyi2012magic}, which is a resource-intensive subroutine with a costly overhead~\cite{bravyi2012magic,Campbell2017e} (see Refs.\cite{Hastings2018,Chamberland2018,Wang2018,OGorman2017,Krishna2018a,Fang2019,Regula2019,wang2019quantifying} for recent progress).

Quantum resource theories \cite{Chitambar2018} offer a powerful framework for studying quantum resources for information processing.
\diff{The resource theory of $T$ count has found applications in magic state distillation and non-Clifford gate synthesis~\cite{beverland2020lower}.}  
To better realize quantum computation, there are many \diff{circuit synthesis} algorithms \diff{that} aim to reduce the number of $T$ gates in the implementation of a given unitary operator ~\cite{selinger2015efficient,ross2016optimal,heyfron2018efficient,gosset2014algorithm,Mosca2020}, providing a better upper bound for its $T$ count. 

This work aims to study lower bounds for the $T$ count. \diff{A lower bound for the $T$ count  estimates the  minimum magic resources~\cite{beverland2020lower} required for quantum computation, and gives a lower bound for the runtime of the classical simulation algorithms~\cite{seddon2021quantifying,bravyi2016improved,bravyi2016trading} based on Clifford+$T$ gate sets.}
Specifically, for a given $n$-qubit unitary $U$, write $\lb{U}$ as the minimum number of $T$ gates among all \diff{implementations} of $U$ over the Clifford+$T$ gate set, i.e., the $T$ count equals $\lb{U}$. The task is to calculate a quantity $v(U)$ such that $v(U)\leq \lb{U}$. A good lower bound is a $v(U)$ which is as large as possible, and the time complexity for computing $v(U)$ is as small as possible. From the theoretical perspective, the explicit expressions of $v(U)$ for some $n$-qubit unitary families are useful and desirable. For the decision version of the lower-bound task, Gosset \emph{et al.}~\cite{gosset2014algorithm} gives an algorithm which for any integer $m$ and $n$-qubit unitary $U$, it decides whether the $T$ count for $U$ is less than $m$, with both time and space complexity of $O(2^{nm}poly(m,2^n))$. Although this algorithm can be used to calculate the exact value of the $T$ count, the time complexity is quite high. Even for small $m$ such as $m=n$, 
%For even linear $m$, such as $m=n$,
 the complexity will scale as $O(2^{cn^2})$ for some constant $c$. Other methods for deriving lower bounds for the $T$ count are mostly based on magic monotones. Those magic monotones are primarily designed for quantifying the non-Clifford resources for constructing quantum states, such as quantifying the number of the magic state $\ket{T}$ needed for constructing a given state $\ket{\psi}$. Those magic monotones could also naturally induce a lower bound for the $T$ count for a given unitary. Note that due to the adaptive operations allowed in the model of computation with magic states, which are not allowed in the model of computation with unitary, the lower bound of the $T$ count for $U$ induced by those magic monotones can be strictly smaller than the $T$ count for $U$~\cite{howard2017application,jones2013low,gosset2014algorithm}. 
This subtle distinction is clarified in Appendix \ref{appendix:relation}.

The relevant magic monotones can be summarized as follows. For a pure state $\ket{\psi}$,
\diff{the} stabilizer rank~\cite{bravyi2019simulation} is defined as the minimum number of nonzero coefficients when decomposing $\ket{\psi}$ onto a linear combination of stabilizer states. The stabilizer extent~\cite{bravyi2019simulation,Regula2017a} is defined as the square of the minimum $l_1$ norm of the above coefficients. It relates to the approximate stabilizer rank, that is the stabilizer rank of states, which approximate $\ket{\psi}$ within given precision. The robustness of  nonstabilizerness~\cite{howard2017application} is also related to the minimum $l_1$ norm but defined for mixed states. The stabilizer rank can also be generalized to mixed states in multiple ways~\cite{seddon2021quantifying}. The channel robustness and magic capacity~\cite{seddon2019quantifying} are closely related to the robustness of magic. They can be used for quantifying the non-Clifford resources for  multiqubit operations. All of the above magic monotones are more or less involving optimization over the set of stabilizer states, and thus are costly to compute
since the cardinality of stabilizer states increases as 
$2^{\Omega(n^2)}$~\cite{gross2006hudson}.
	%$2^{(1/2+o(1))n^2}$~\cite{aaronson2004improved}. 
Moreover, since they involve optimization, it is hard to derive explicit expressions for those magic monotones for $n$-qubit unitaries or quantum states. Beyond the magic monotones based on optimization, Beverland \emph{et al.}~\cite{beverland2020lower} proposed the state-stabilizer nullity based on quantities related to the Pauli group. More specifically, let $s(\ket{\psi})$ be the size of the stabilizer of $\ket{\psi}$, where the stabilizer is the subgroup of the Pauli group for which $\ket{\psi}$ is a $+1$ eigenvector, the state-stabilizer nullity of $\ket{\psi}$ is defined as $v_s(\ket{\psi})=n-\log_2s(\ket{\psi})$. In general, the state-stabilizer nullity can be computed in time $O(2^{4n})$. 
Moreover, for some useful unitary families $\{U_n\}_n$, the explicit expressions of the state-stabilizer nullity of  $U_n\ket{+}^{\otimes n}$ can be easily computed. And this quantity can be used to lower bound the $T$ count for $U_n$.
 An example is showing that the $n$-qubit Toffoli gate has $T$ count $n$, 
 implying that implementing such a gate needs at least $n$ $\ket{T}$ states or $T$ gates. The state-stabilizer nullity is designed for quantifying non-Clifford resources for constructing states, although it can be used to derive lower bounds for the $T$ count for unitaries, a ``unitary-stabilizer nullity" is desirable for better quantifying the nonstabilizerness and estimate the $T$ count of multiqubit gates.
It is unclear and nontrivial to generalize the definition from states to unitaries. Beyond all the above, there are also magic monotones defined for qudit systems of odd prime dimension, which cannot be directly used for qubit systems~\cite{emerson2014resource,wang2019quantifying}.

In this work, we define the unitary-stabilizer nullity. Instead of counting the operators in the Pauli group,  which stabilize $\ket{\psi}$, for a target $U$  we count the $\pm 1$s in the Pauli transfer matrix. 
In other words, we calculate the size of the intersection between $U\cP_n U^\dagger$ and $\cP_n$, where  $\cP_n$ is the quotient Pauli group, that is   the Pauli-group modulo phases.
Denoting the size as $s(U)$, the unitary-stabilizer nullity is defined as $v(U)=2n- \log_2s(U)$.
Intuitively, $s(U)$ quantifies the similarity between $U$ and Clifford unitaries, where for Clifford unitary $U$, $s(U)=4^n$.
To use $v(\cdot)$ to derive a lower bound for the $T$ count, we prove two key properties: $v(\cdot)$ satisfies faithfullness, i.e. $v(U)\geq 0$ for any $U$, $v(U)=0$ if and only if $U$ is Clifford,  and subadditivity under composition,  i.e. $v(U\cdot V)\leq v(U)+v(V)$ or equivalently, $s(U)s(V)\leq 4^n s(U V)$. We further apply these properties to establish the lower bound for the $T$ count.

We then compare the lower bounds of the $T$ count obtained by the state-stabilizer nullity and the unitary-stabilizer nullity. We show that the bound given by the unitary-stabilizer nullity will never be smaller than the bound given by the state-stabilizer nullity, that is $v(U)\geq v_s(U|\phi\rangle),$ for any stabilizer state $\ket{\phi}$. 
We also prove that for any diagonal $U$, $v(U)= v_s(U|+\rangle^{\otimes n})$, which in turn shows that for any diagonal $U$, the maximum lower bound obtained by the state-stabilizer nullity, i.e.,  $v_s(U|\phi\rangle)$ for some stabilizer state $\ket{\phi}$, is achieved by $|\phi\rangle=|+\rangle^{\otimes n}$.
Moreover, we give an explicit $n$-qubit unitary family, which has unitary-stabilizer nullity $2n$. Since the state-stabilizer nullity is always less than $n$, such a unitary family illustrates that  the lower bound of the $T$ count given by the unitary-stabilizer nullity can be strictly better than the bound given by the state-stabilizer nullity by a factor of $2$. 

In general, the unitary-stabilizer nullity can be computed in time $O(2^{7n})$, which is less than $O(2^{(1/2+o(1))n^2})$ but still costly. However, as mentioned above, for certain $n$-qubit unitary families the explicit expressions of the unitary-stabilizer nullity can be easily computed. We believe this property will make it useful in the theoretical study. At the technical level, most of the proofs are based on connecting $s(U)$ to certain subgroups and applying the Lagrangian theorem~\cite{Dummit2004} and its variants, deriving relationships between the sizes of subgroups. Based on the group theory, an interesting fact is that unlike other magic monotones, both the state- and the unitary-stabilizer nullity are always integers. Note that since the unitary-stabilizer nullity is bounded by $2n$, the lower bound for $T$ count it gives is up to $2n$.

Finally, we connect the unitary-stabilizer nullity and the state-stabilizer nullity with auxiliary systems. Specifically, we prove that $v(U)$ equals the best bound one can obtain by the state-stabilizer nullity with auxiliary systems. That is for any $d$-qubit auxiliary system, for any $(d+n)$-qubit stabilizer state $\ket{\psi}$, we have $v(U)\geq v_s((I_{2^d}\otimes U \ket{\psi})$. In particular, we show that
$v(U)=v_s((I_{2^n}\otimes U)\ket{\phi})$ where $\ket{\phi}$ is the $2n$-qubit maximally entangled state.  We also give numerical results that \diff{compare} the $T$ count lower bound obtained by the state-stabilizer nullity, the unitary-stabilizer nullity, and the stabilizer extent. The results \diff{show} that for several gates of practical interest, the unitary-stabilizer unitary can derive \diff{a} better $T$ count lower bound.

\paragraph*{Outline and main contributions}
The outline and main contribution of this paper can be summarized as follows:
\begin{itemize}
    \item In section \ref{sec:prelim}, we set the notations and  introduce the unitary-stabilizer nullity.
    \item In section \ref{sec:basic}, we prove basic properties of the unitary-stabilizer nullity, that is  faithfulness, invariance under Clifford operations, and additivity under tensor product. Moreover, we show that the unitary-stabilizer nullity satisfies the subadditivity under composition. 
    \item  In section \ref{sec:lowerT}, we show that the unitary-stabilizer nullity can derive a lower bound of the $T$ count for a given unitary $U$. \diff{The method can be generalized to give lower bounds for arbitrary non-Clifford gates beyond  the $T$ gate.}
    We also compare the lower bound  obtained by the unitary- and the state-stabilizer nullity.  We prove that the bound by the unitary-stabilizer nullity  will never perform worse than the state-stabilizer nullity. Moreover, we  show that there is an explicit $n$-qubit unitary  family of unitary-stabilizer nullity $2n$, illustrating our bounds can strictly  outperform the state-stabilizer nullity. Besides, for diagonal unitaries, we prove  our  bounds coincide with the state-stabilizer nullity. Then  connect the unitary-stabilizer nullity with the state-stabilizer nullity with auxiliary systems. 
    \item In section~\ref{subsec: examples}, we investigate the $T$ counts of a plethora of quantum gates of practical interest and the improvements are evident for many cases.
    \item Finally in Section \ref{sec:open}, we \diff{conclude} and discuss open problems. 
\end{itemize}

%%%%%%%%%%%%%%%%%%%%%%%%%%%%%%%%%%%%%%%%%%%%%%%%
%%%%%%%%%%%%%%%%%%%%%%%%%%%%%%%%%%%%%%%%%%%%%%%
\section{Preliminaries}\label{sec:prelim}
\subsection{Notations}\label{sec:notation}

In this part, we set the notations.  We use $i$ as the imaginary number. We use $\mathbb{C}$ as the complex field, $\mathbb{N}^+$ as the positive integers.
For a complex number $\lambda$, we use $\lambda^\dagger$ as its conjugate, $|\lambda|$ as its norm.  Let $A,B$ be two subgroups of a group $G$, we use $A\times B$ as the set $A\times B:=\{ab|a\in A, b\in B\}$. Note that when $G$ is Abelian, $A\times B$ is a subgroup of $G$. For any $n$-bit string $\bx\in\{0,1\}^n$, we use $\odot(\bx)$ as the product $x_1x_2...x_n$. We use $1^n$ as the $n$-bit string of all ones. Given two strings $\bx,\by\in\{0,1\}^n,$ we use $\bx\cdot \by$ as their inner product,  $\bx\oplus \by$ as their bitwise XOR operation.
 We abbreviate $x\in\{-1,+1\}$ as $x=\pm 1$. We use $I_d$ as the $d$-dimensional identity matrix. We abbreviate $I_d$ as $I$ when $d$ is clear in the context. We write the state $\frac{1}{\sqrt{2}}(\ket{0}+\ket{1})$ as $\ket{+}$.
For a matrix $U$,
we use $U^T$ as its transpose, $U^\dagger$ as its transpose conjugate. For integers $d$ and $k$, we use $d|k$ to denote $d$ is a factor of $k$. We use the $\delta_{\bu\bv}$ to denote the function that equals $1$ if and only if $\bu=\bv$ and $0$ otherwise.

We use $[n]$ as the set $\{0,1,...,n-1\}.$ We use $4^{[n]}$ as the set of length-$n$ vectors with alphabet $[4]$, i.e. $4^{[n]}:=\{0,1,2,3\}^n$. We use \diff{the} lowercase character to denote a number, and \diff{the} bold lowercase character to denote a vector, that is $u\in [4]$ and $\bu \in [4]^n$ respectively.
We represent Pauli operators as 
\begin{align}
    &\sigma_0=I_2=\begin{bmatrix}
    1&0\\
    0&1
    \end{bmatrix},
    \sigma_1=X=\begin{bmatrix}
    0&1\\
    1&0
    \end{bmatrix},\\
    &\sigma_2=Y=\begin{bmatrix}
    0&-i\\i&0
    \end{bmatrix},
    \sigma_3=Z=\begin{bmatrix}
    1&0\\0&-1
    \end{bmatrix}.
    \end{align}
    We write phase gate\,(S), $T$ gate\,(T) and Hadamard gate\,(H) as
\begin{align}
    S=\begin{bmatrix}
    1 &0\\0 &i
    \end{bmatrix},
    T=\begin{bmatrix}
    1 &0\\0 &e^{\frac{\pi i}{4}}
    \end{bmatrix},
    H=\frac{1}{\sqrt{2}}\begin{bmatrix}1
    &1\\1&-1\end{bmatrix}.
\end{align}
 The CNOT gate is the two-qubit control NOT gate, i.e. $CNOT|x_1\rangle|x_2\rangle=|x_1\rangle|x_1\oplus x_2\rangle$. The $C^{n-1}Z$ is the multicontrolled Z  gate, i.e. $C^{n-1}Z|x_1...x_n\rangle=(-1)^{x_1x_2\cdots x_n}|x_1...x_n\rangle$.

The one-qubit Pauli group is the group generated by the Pauli operators, i.e. $\hat{\cP}_1=\{ i^k \sigma_h | k,h\in [4]\}$. The $n$-qubit Pauli group is defined to be the $n$-fold tensor product of the one-qubit Pauli-group, i.e. $\hat{\cP}_n:=\hat{\cP}_1^{\otimes n}$. 
In the following, we are uninterested in the overall phases. We define the one-qubit Pauli-group modulo phases to be 
$\cP_1=\hat{\cP}_1/\{i^k I_2|k\in[4]\}=\{\sigma_h | h\in[4]\}$,
where the last equality means we use $\sigma_h$ as a representative of elements in $\{i^k \sigma_h|k\in[4]\}$.
Similarly, we define the $n$-qubit Pauli group modulo phases, or the $n$-qubit quotient Pauli group, to be $\cP_n :=\cP_1^{\otimes n}=\{\sigma_h |h\in [4] \}^{\otimes n}$.
Notice that $\cP_n$ is an Abelian group with respect to matrix-product modulo phases. We have $|\cP_1|=4$
and $|\cP_n|=4^n$ where $|\cdot|$ is the size of the group.
Elements in $\cP_n$ can be represented by a string $\bu\in 4^{[n]}$, that is 
$\sigma_\bu=\otimes_{k=1}^n \sigma_{\bu_k}$.

The $n$-qubit Clifford group $\cC_n$ is the set of $n$-qubit unitary $U$ which maps Pauli group to Pauli group, i.e. $U\hat{\cP}_n U^\dagger =\hat{\cP}_n$.
An $n$-qubit state $\ket{\psi}$ is called a stabilizer state if $\ket{\psi}=U\ket{0}^{\otimes n}$ for some $U\in \cC_n$.
One can \diff{verify} that $U\cP_n U^\dagger =\cP_n$, with respect to matrix-product modulo phases. In the following context, when involving the quotient group $\cP_n$, we always assume the operation is matrix product modulo phases.

 For a given unitary $U$, we define its $T$ count $t(U)$  as the minimum number of $T$ gates required when decomposing $U$ into a sequence of gates over  Clifford plus $T$ gate set, without auxiliary qubits and measurement.  Here we view Clifford gates, or Clifford unitaries, as free resources. More clarifications can be seen in the quantum computation with unitary part in Appendix \ref{appendix:relation}. If $U$ cannot be exactly synthesized by  Clifford plus $T$ gate set, we denote its $T$ count as $+\infty$.
\subsection{Pauli function}
In this part, we introduce the state Pauli function and the unitary Pauli function, which are the basic elements for defining the state- and the unitary-stabilizer nullity. 
Those functions are related to the values of the matrix representation of a linear map in the Pauli basis. We adopt the terminology from Ref.~\cite{beverland2020lower}.

We name them the Pauli functions  in order to unify the framework of the state- and the unitary-stabilizer nullity, and to better analyze their properties.

\begin{definition}[State Pauli function]
 For $\bu\in 4^{[n]}$,
the Pauli function of an $n$-qubit pure state $\ket{\psi}$ is defined as
$P_{\ket{\psi}}(\bu):= \tr \proj\psi \sigma_{\bu}$.
\end{definition}
 
\begin{definition}[Unitary Pauli function] For $\bu,\bv\in 4^{[n]}$, 
 the Pauli function of an $n$-qubit unitary $U$ is defined as
$P_{U}(\bu|\bv) =  \tr (\sigma_{\bu}U\sigma_{\bv}U^\dagger)/2^n$.

\end{definition}
The matrix whose $(\bu,\bv)$ th element is $P_U(\bu|\bv)$ is known to be the Pauli transfer matrix of $U$. The Pauli transfer matrix is the matrix representation of  the linear map $U(\cdot)U^\dagger$ in the Pauli basis. 
Equivalently, one can show that
$U\sigma_\bv U^{\dagger}= \sum_\bu P_U(\bu|\bv)\sigma_\bu.$
A useful fact is \diff{that} the Pauli transfer matrix of unitary is orthogonal. That is for unitary $U$, 
\begin{align}\label{eq:orthogonal}
	\!\!\!\sum_{\bu} P_U(\bu|\bv)P_U(\bu|\bv')= \tr(U\sigma_{\bv}U^\dagger U\sigma_{\bv'}U^\dagger)/2^n=\delta_{\bv\bv'}.
\end{align}

Both the state and the unitary Pauli functions take values in $[-1,1]$.
To analyze the properties of the unitary Pauli function, we connect $U$ to a subgroup of the quotient Pauli group.
\begin{definition}[Pauli subgroup associated with $U$]
Given an $n$-qubit unitary $U$, we define the Pauli subgroup associated with $U$ to be
$\cP_U := U\cP_n U^\dagger \cap \cP_n$.
\end{definition}
One can verify that $\cP_U $ is a subgroup of $\cP_n$. It is worth noting that $U\cP_U U^\dagger$ may \textbf{not} be equal to $\cP_U$, and $\cP_U$  is \textbf{not} $U\cP_nU^\dagger$.

\subsection{Stabilizer nullity}\label{sec:nullity}
In this part, we give  the definitions of the state-  and  the unitary-stabilizer nullity.

\begin{definition}[state-stabilizer nullity~\cite{beverland2020lower}]
\label{def:ssn}
For an $n$-qubit state $\ket \psi$, the stabilizer nullity of $\ket \psi$ is defined as
$v_s(\ket\psi):=n-\log_2 s(\ket{\psi})$,
where $s(\ket{\psi})$ is equal to the number of $\pm$1s in the Pauli function of $\ket{\psi}$ when varying $\bu\in 4^{[n]}$. 
\end{definition}
% By  definition, the state-stabilizer nullity is at most $n$.

To better explore the power of stabilizer nullity in quantifying the $T$ count (or nonstabilizerness) of quantum dynamics, we introduce the unitary-stabilizer nullity as follows.
\begin{definition}[unitary-stabilizer nullity]\label{def:usn}
For an $n$-qubit unitary $U$, the stabilizer nullity of $U$ is defined as
\begin{align}
v(U) := 2n - \log_2 s(U).
\end{align}
where $s(U)$ is equal to  the number of $\pm 1$s in the Pauli function of $U$ when varying $\bu,
\bv\in 4^{[n]}$.
\end{definition}

Note that Beverland \emph{et al}~\cite{beverland2020lower} 
originally define the state-stabilizer nullity as $\nu(\ket{\psi})=n-\log_2{|\mathrm{Stab}\, \ket{\psi}}|$, where  $\mathrm{Stab}\, \ket{\psi}=\{P\in \hat{\cP}_n: P|\psi\rangle=|\psi\rangle\}$, 
that is $\mathrm{Stab}\, \ket{\psi}$ is the subgroup of the Pauli group $\hat{\cP}_n$ for which  $\ket{\psi}$ is a $+1$ eigenvector. 
Definition \ref{def:ssn} is equivalent to Beverland's definition, that is $v_s(\ket{\psi})=\nu(\ket{\psi})$, as illustrated in Appendix \ref{appendix:equ_state_stab}. 

% Examples to help readers to understand the definitions.
For any stabilizer state $\ket{\psi}$, Beverland \emph{et al}~\cite{beverland2020lower} show that $v_s(\ket{\psi})=0$. For any Clifford unitary $U$, one can verify that  $v(U)=0$ by
Lemma \ref{lem:Clifford10}.

\section{Basic properties of  unitary-stabilizer nullity}\label{sec:basic}

\subsection{Basic properties of the unitary Pauli function}
We firstly analyze the basic properties of the unitary Pauli function. 
Roughly speaking, the unitary Pauli function measures the similarity between $U$ and the Clifford operators, where the similarity is quantified by  how well the map $U(\cdot)U^{\dagger}$  keeps the Pauli group invariant. To see this intuition, we prove some basic properties of   the unitary Pauli function. 

We give short proofs to Lemma \ref{lem:Clifford10}, Lemma \ref{lem:10s} and Lemma \ref{lem:congs} assuming the familiarity of quantum information. The more detailed proofs are put into Appendix \ref{appendix:proof_Lemma_12}.  

\begin{lemma}\label{lem:Clifford10}
If 
 $U$ is an $n$-qubit Clifford unitary, for $\bu,\bv\in 4^{[n]}$,
 \begin{itemize}
 \item for every $\bu$, there is only one $\bv$ such that $P_{U}(\bu|\bv)=\pm 1$, and for all $\bv'\neq \bv$, $P_{U}(\bu|\bv')=0$.
     \item for every $\bv$,    there is only one $\bu$ such that $P_{U}(\bu|\bv)=\pm 1$, and for all $\bu'\neq \bu$, $P_{U}(\bu'|\bv)=0$. 
 \end{itemize}
 \end{lemma}
 
 \begin{proof}
 We  give only proof to the first statement, the second is similar. Recall that
 \begin{align}
      P_U(\bu|\bv) &= \tr(\sigma_\bu U\sigma_\bv U^\dagger )/2^n \label{eq:12},\\
     \tr(\sigma_\bu \sigma_\bv) &= \left\{
     \begin{aligned}\label{eq:c01}
     & 2^n, & \text{ if $\bu=\bv$}\\
     & 0, & \text{ else.}
     \end{aligned}
     \right.
 \end{align}
 
 Since $U$ is Clifford, $U(\cdot)U^\dagger$ is an isomorphism of $\hat{\cP}_n$, i.e.  $U \hat{\cP}_n U^\dagger= \hat{\cP}_n$. In other words, $U(\cdot)U^\dagger$ is a permutation over the Pauli operators $\{\sigma_\bu\}_{\bu \in 4^{[n]}}$, up to a $\pm1$ sign. Thus by Eq.~\eqref{eq:12} and Eq.~\eqref{eq:c01}, we know 
 $P_U(\bu|\bv)=\pm 1$ if and only if $U\sigma_{\bv}U^\dagger =\pm \sigma_\bu$. Thus the lemma holds.	
 \end{proof}

Similar results also hold for general unitary $U$ beyond Clifford operators, which is proved in the following lemma.

\begin{lemma}\label{lem:10s}
For any $n$-qubit unitary $U$,
if there exist $\bu,\bv\in 4^{[n]}$ such that $P_U(\bu|\bv)=\pm 1$. Then 
\begin{align}
&U\sigma_\bv U^\dagger=\pm \sigma_\bu,\\
    &\forall \bu'\neq \bu, P_U(\bu'|\bv)=0.
\end{align}
\end{lemma}
\begin{proof}
Recall that from equation (\ref{eq:orthogonal})
we know the Pauli transfer matrix of $U$, whose $(\bu,\bv)$-th element is $P_U(\bu,\bv)$, is orthogonal. 
Specially, choosing $\bv'=\bv$ in Eq.~\eqref{eq:orthogonal} we know 
\begin{align}
	\sum_{\bu} |P_U(\bu|\bv)|^2= 1. 
\end{align}
 Thus if $P_U(\bu|\bv)=\pm 1$,  we have $\forall \bu'\neq \bu, P_U(\bu'|\bv)=0.$ Since $U\sigma_\bv U^{\dagger}= \sum_\bu P_U(\bu|\bv)\sigma_\bu,$ thus we complete the proof.

\end{proof}

In the following, we give a simple but useful connection between the unitary-stabilizer nullity and the subgroup associated with $U$. Recall that $s(U)$ is the number of $\pm 1$s in the unitary Pauli function, as in Definition \ref{def:usn}.
In the following lemma, we prove that $s(U)$ equals to the size of $\cP_U$, that is the intersection of $U\cP_n U^\dagger$ and $\cP_n$. 
Note that  $\cP_n$  is the quotient Pauli group  rather than  the Pauli group $\hat{\cP_n}$.

\begin{lemma}\label{lem:congs}
For any $n$-qubit unitary $U$, $\cP_U$ is a subgroup of $\cP_n$. What is more,
\begin{align}\label{eq:psu}
    |\cP_U|=|\cP_{U^\dagger}|=s(U).
\end{align}
\end{lemma}
\begin{proof}
By Lemma \ref{lem:10s}, we know every $\pm1$ in the Pauli transfer matrix will imply a distinct element in $\cP_U$. That is if $P_U(\bu_{\bv}|\bv)=\pm 1$, then $\sigma_{\bu_{\bv}}\in \cP_U$.  Besides, for different $\bv$, such $\bu_{\bv}$ are different, this comes from the fact that the Pauli transfer matrix is orthogonal, thus any row cannot have two $\pm1$s. Thus
$|\cP_U|\geq |s(U)|$. On the other hand, for every element in $\cP_U$, that is a $\sigma_\bu$ such that $U\sigma_\bv U^\dagger =\pm\sigma_\bu$ for some $\bv$, we would have
     $|P_U(\bu|\bv)|=1$, thus $|\cP_U|\leq |s(U)|$.
Thus $|P_U|=s(U)$. Besides, we have $\tr(\sigma_\bu U \sigma_\bv U^\dagger)=\tr(\sigma_\bv U ^\dagger\sigma_\bu U)$,
thus the Pauli transfer matrix of $U^\dagger$ is the transpose of the Pauli transfer matrix of $U$. Thus
$|\cP_{U^\dagger}|=|s(U)|=|\cP_{U}|$. 
\end{proof}

An interesting corollary is that the unitary-stabilizer nullity is  always an integer, although we do not utilize this property in the following sections.
\begin{corollary}
\label{lem:vinteger}
For any $n$-qubit unitary $U$, $v(U)$ must be an integer, or equivalently, $s(U)=2^k$ for some non-negative integer $k$.
\end{corollary}
By Lemma \ref{lem:congs}, we know $s(U)=|\cP_U|$.
Since $\cP_U$ is a subgroup of $\cP_n$, where $|\cP_n|=4^n$,  by group theory, we know that $s(U)|4^n$.
Thus $v(U)=2n-\log_2 s(U)$
is always an integer.

\subsection{Basic properties of the unitary-stabilizer nullity }
 
In this section, we  prove the basic properties of the unitary-stabilizer nullity, that is the faithfulness, invariance under Clifford operations,  
additivity under tensor product and subadditivity under composition.

\begin{theorem}[Faithfulness]\label{prop:faith}
Let  $U$ be an $n$-qubit unitary. Then $v(U)\ge 0$. Furthermore,  $v(U)= 0$ if and only if $U$ is a Clifford unitary. Or equivalently, $s(U)\leq 4^n$ and $s(U)=4^n$ if and only if $U$ is a Clifford unitary.
\end{theorem}
\begin{proof}
For any $n$-qubit unitary $U$, we firstly prove $v(U)\ge 0$, or equivalently, $s(U)\leq 4^n$. By Lemma \ref{lem:10s}, if there exist $\bu,\bv$ such that $P_U(\bu|\bv)=1$ then for any $\bu'\neq \bu, \cP_U(\bu'|\bv)=0$. Thus every $\pm 1$ in the Pauli function must come together with $4^n-1$ number of  $0s$. Thus $s(U)\leq 4^n$.

If $s(U)=4^n$, by  Lemma \ref{lem:10s}, for every Pauli operator $\sigma_\bv$, $U \sigma_\bv U^\dagger$ is a Pauli operator, thus $U$ is a Clifford unitary. On the other hand, if $U$ is a Clifford operator, by Lemma \ref{lem:Clifford10}, it is easy to verify $s(U)=4^n.$
\end{proof}

\begin{theorem}[Invariant under Clifford operation]\label{thm:invClifford}
For  an $n$-qubit unitary $U$ and an $n$-qubit Clifford unitary $V$, 
\begin{align}
v(V U) = v(U),\quad s(V U) = s(U).
\end{align}
Similarly, it also holds that $v(U V) = v(U),s(U V) = s(U)$. 
\end{theorem}
\begin{proof}
Recall that the unitary Pauli function
\begin{align}
    P_{V U}(\bu|\bv) &=\tr \sigma_\bu VU \sigma_\bv U^\dagger V^\dagger/2^n\nonumber\\
     &=\tr   V^\dagger \sigma_\bu VU \sigma_\bv U^\dagger/2^n.
\end{align}
Since $V$ is Clifford, then $V^\dagger (\cdot) V$ is an isomorphism over the Pauli group,  i.e. $V^\dagger \hat{\cP}_n V=\hat{\cP}_n$. Further, since $V^\dagger \sigma_\bu V$ does not change the eigenvalue of $\sigma_\bu$, where is $\pm 1$. Thus 
\begin{align}
    V^\dagger \{\sigma_\bu|\bu \in 4^{[n]}\} V=\{p_\bu \sigma_\bu|\bu \in 4^{[n]}, p_\bu=\pm 1\}.
\end{align}
In other words, the set of unitary Pauli functions satisfies
\begin{align}
    \{P_{V U}(\bu|\bv)\}=\{p_\bu P_{U}(\bu|\bv)\, | \, p_\bu=\pm 1\}.
\end{align}
Thus $s(V U)=s(U)$, thus $v(V  U) = s(U)$. Similarly, it  holds that $v(V U) = v(U), s(V U) = s(U)$.
\end{proof}

\begin{theorem}[Additivity under tensor product]\label{thm:addten}
Given $n$-qubit unitary $U$ and $m$-qubit unitary $V$, the following holds
\begin{align}
    v(U\otimes V)=v(U)+v(V).
\end{align}
\end{theorem}

\vspace{0.2em}

\begin{proof}
For simplicity, given $\bu\in 4^{[n+m]}$, we write $\bu^{(1)}$ as the first $n$ bits of $\bu$, i.e. $\bu^{(1)}=\bu_1\bu_2...\bu_n$. We write $\bu^{(2)}$ as the last $m$ bits of $\bu$,  i.e. $\bu^{(2)}=\bu_{n+1}\bu_{n+2}...\bu_{n+m}$. 
It suffices to notice that 

\begin{align}
  &P_{U\otimes V}(\bu|\bv) \nonumber\\
  &=\tr \sigma_{\bu^{(1)}}\otimes \sigma_{\bu^{(2)}} \left(U\otimes V\right)  \sigma_{\bv^{(1)}}\otimes \sigma_{\bv^{(2)}} \left( U^\dagger \otimes V^\dagger \right)  / 2^{(n+m)} \nonumber\\
  &= \tr \sigma_{\bu^{(1)}} U \sigma_{\bv^{(1)}} U^\dagger  / 2^n \cdot \tr \sigma_{\bu^{(2)}}V   \sigma_{\bv^{(2)}}  V^\dagger   / 2^m \nonumber\\
    &=P_U(\bu^{(1)}|\bv^{(1)})\cdot P_V(\bu^{(2)}|\bv^{(2)}).
\end{align}
Thus we have
$s(U\otimes V)=s(U)s(V)$,
$2(n+m)-\log_2s(U\otimes V)=(2n-\log_2s(U))+(2m-\log_2s(V))$,
and then $v(U\otimes V)=v(U)+v(V)$.
\end{proof}

 \begin{lemma}\label{lem:leqsuv}
 For any $n$-qubit unitary $U$ and $V$,
 \begin{align}
     |\cP_{U^\dagger}\cap \cP_V|\leq s(UV).
 \end{align}
 \end{lemma}
Recall that 
\begin{align}
    \cP_{UV} &= UV \cP_n V^\dagger U^\dagger \cap \cP_n\nonumber\\  
    &= U\left(V\cP_n V^\dagger \cap U^\dagger \cP_n U \right) U^\dagger.
\end{align}
Noting that $U(\cdot)U^\dagger$ is an injection,  we have 
\begin{align}
    |\cP_{UV}| &=|V\cP_n V^\dagger \cap U^\dagger \cP_n U| \nonumber\\
    &\geq |V\cP_n V^\dagger \cap U^\dagger \cP_n U\cap \cP_n|\nonumber\\
    &\geq |\cP_V\cap \cP_{U^\dagger}|.
\end{align}
By Lemma \ref{lem:congs} we know $|\cP_{UV}|=s(UV)$, we conclude
$|\cP_{U^\dagger}\cap \cP_V|\leq s(UV)$.

The following lemma is a variation of Lagrange's theorem in group theory. It connects the sizes between subgroups and their intersections. A detailed proof is provided in Appendix.~\ref{appendix:proof_Lemma_12}.

\begin{lemma}\label{lem:pj}
Suppose $A,B$ are subgroups of $H$, then
$|A\times B|=\frac{|A|\cdot|B|}{|A\cap B|}$.
Since $A\times B\subseteq H$ we have
$|A|\cdot |B|\leq |H|\cdot |A\cap B|$. 
\end{lemma}

Utilizing Lemma \ref{lem:pj}, we finally arrive at the statement that the unitary-stabilizer nullity satisfies the subadditivity under composition.  
\begin{theorem}[Subadditivity under composition] \label{thm:comp}
Given $n$-qubit unitary $U$ and $V$, the following holds
\begin{align}
    v(U V)\leq v(U)+v(V).
\end{align}
Or equivalently
$s(U)s(V)\leq 4^n s(U V)$.
\end{theorem}
\begin{proof}
Let $H=\cP_n, A=\cP_{U^\dagger}, B=\cP_V$, we know $A,B$ are subgroups of $H$. By Lemma \ref{lem:pj} we have 
\begin{align}
 |\cP_{U^\dagger}| \cdot |\cP_V| \leq  |\cP_{U^\dagger} \cap \cP_V|\cdot  4^n.
\end{align}
Combine with Lemma \ref{lem:leqsuv}, we have
\begin{align}
  |\cP_{U^\dagger}| \cdot |\cP_V|\leq   |s(UV)|\cdot  4^n.
\end{align}
Further by Lemma \ref{lem:congs}, we know 
\begin{align}
    |\cP_V|=s(V), |\cP_{U^\dagger}|=|\cP_{U}|=s(U).
\end{align}
thus finally we get
\begin{align}
    s(U)s(V)\leq 4^n s(U V).
\end{align}
\end{proof}

\section{Applications}\label{sec:lowerT}
Interestingly, by the properties in section \ref{sec:basic}, especially the subadditivity under composition, 
 we can use $v(\cdot)$ to give a lower bound for the $T$ count.  In section \ref{sec:llt}, we prove this connection. In section \ref{sec:compsu}
 we compare the lower bound for $T$ count obtained by the state- and the unitary-stabilizer nullity, and prove that the unitary-stabilizer nullity never performs worse than the state-stabilizer nullity. We also prove that for diagonal unitary, the two lower bounds coincide.  Finally, in section \ref{sec:expli}, we give an $n$-qubit unitary family which has unitary-stabilizer nullity $2n$. This circuit family  gives an example where the lower bound obtained by the unitary-stabilizer nullity strictly outperforms the state-stabilizer nullity by a factor of $2$.
 Furthermore, in section \ref{sec:aux_better}, we connect the unitary-stabilizer nullity and the state-stabilizer nullity with auxiliary systems, showing that adding auxiliary systems and choosing proper stabilizer states can strictly \diff{improve} the lower bound obtained by the state-stabilizer nullity.

 \subsection{Lower bounds for the $T$ count}
\label{sec:llt}
Recall that for a given unitary $U$, its $T$ count $\lb{U}$ is the minimum number of $T$ gates used among any implementation of $U$ that decomposes $U$ into a sequence of gates from Clifford plus $T$-gate set, without ancillary qubits and measurement. In this section, we would use the unitary-stabilizer nullity to lower bound the $T$ count. The main idea is to use the fact that $v(\cdot)$ satisfies  the subadditivity under composition.
 
\begin{theorem}[Lower bound for the  $T$ count]\label{thm:lbt}
For any $n$-qubit unitary $U$, its $T$ count is lower bounded by its unitary-stabilizer nullity, that is
\begin{align}
    \lb{U} \geq v(U).
\end{align}
\end{theorem}

\begin{proof}
Notice that $v(T)=1$ by verifying the unitary Pauli function of $T$
\begin{align}
 &\{P_T(\bu|\bv)\}=
 \begin{bmatrix}
 1&0&0&0\\
    0&\frac{\sqrt{2}}{2}&-\frac{\sqrt{2}}{2}&0\\
        0&\frac{\sqrt{2}}{2}&\frac{\sqrt{2}}{2}&0\\
    0&0&0&1
\end{bmatrix}.
\end{align}
Thus, we have $v(T)=2-\log_2 2=1$.

Besides, since $I_2$ is a Clifford operator, by the faithfulness, i.e. Theorem \ref{prop:faith}, we have 
\begin{align}
    v(I_2)=0.
\end{align}
By the additivity under tensor product, i.e. Theorem \ref{thm:addten}, we have 
\begin{align}
    v(T\otimes I_2 \otimes I_2 ... \otimes I_2)=v(T)+0=1.\label{eq:29}
\end{align}
\vspace{0.5em}
From Theorem \ref{prop:faith} we know that all the Clifford unitaries have unitary-stabilizer nullity $0$.
Suppose $U$ is decomposed into sequence of Clifford gates and $\lb{U}$ $T$ gates,
using the subadditivity of composition sequentially, that is Theorem \ref{thm:comp}, we have
\begin{align}
    v(U)&\leq 0+\lb{U}\cdot v(T) \label{eq:subadd}\\
    &=\lb{U}\nonumber.
\end{align}
\diff{This} concludes the theorem.
\end{proof}

To compare the lower bound got by the unitary-stabilizer nullity and the lower bound by the state-stabilizer nullity, we would use the following theorem from \cite{beverland2020lower}.
\begin{theorem}\cite{beverland2020lower}\label{thm:state_lower}
For any $n$-qubit unitary $U$, for any $(n+t)$-qubit stabilizer state $\ket{\phi}$, the $T$ count of $U$ is lower bounded by
\begin{align}
	t(U)\geq v_s\left( (I_{2^t}\otimes U)\ket{\phi}\right).
\end{align}
\end{theorem}

To generalize our result a little bit, \diff{consider} the task that decomposing $U$ into sequence of Clifford gates and $W$ gates, where $W$ gate  is a gate beyond Clifford. \diff{Using} similar proofs as in the Theorem \ref{thm:lbt}, we can also use the unitary-stabilizer nullity to lower bound the number of $W$ gates.
\begin{theorem}[Lower bound for gate synthesis]\label{thm:general}
For any $n$-qubit unitary $U$, the number of gates $W$ required to realize $U$ under Clifford unitaries, \diff{denoted as $w(U)$}, \diff{satisfies}
\begin{align}
    \diff{w(U)}\geq \frac{v(U)}{v(W)}.
\end{align}
\end{theorem}

\diff{
\begin{proof}
Note that the subadditivity under composition~(Theorem \ref{thm:comp}) holds for arbitrary two unitaries. Suppose $U$  is decomposed into Clifford and $W$ gates, using Theorem \ref{thm:comp} we get similar results as Eq.~(\ref{eq:29}) and Eq.~(\ref{eq:subadd}), i.e.
$v(W\otimes I_2 \otimes I_2 ... \otimes I_2)=v(W)$ and $v(U)\leq w(U)v(W)$. Thus we prove the theorem.
\end{proof}
}

\subsection{Comparison between the unitary and state-stabilizer nullity}\label{sec:compsu}

Beverland \emph{et al.}  \cite{beverland2020lower} showed that
for a unitary $U$, for any stabilizer state $\ket{\psi}$, the state-stabilizer nullity can derive a lower bound for $T$ count for $U$ as $v_s(U\ket{\psi})$.
We compare this lower bound with the lower bound obtained by  the unitary-stabilizer nullity from Theorem \ref{thm:lbt}. 
We prove in Theorem \ref{lem:compa} that our method will never perform  worse and provide the detailed proof in Appendix~D.
We further show that our bound can perform strictly better for a certain unitary family by a factor of 2, as in  Corollary \ref{cor0:2n}. We also prove that for diagonal unitary, the lower bounds for $T$ count obtained by the state- and the unitary-stabilizer nullity are the same, as in Theorem \ref{theo:diag}. 

\begin{theorem}\label{lem:compa}
For any $n$-qubit unitary $U$, for any $n$-qubit stabilizer state $\ket{\psi}$,
\begin{align}
    v(U)\geq v_s(U \ket{\psi}).
\end{align}
That is
\begin{align}
    v(U)\geq \max_{C\in \cC_n} v_s(UC\ket{0}^{\otimes n}).
\end{align}
The equality holds if and only if there exists Clifford $C\in \cC_n$ such that
\begin{align}\label{eq:diageqc}
    &\left\{\bu\in 4^{[n]} \,|\, \langle 0^{\otimes n}|C^{\dagger} U^\dagger \sigma_\bu U C|0\rangle^{\otimes n}=\pm 1 \right\}\nonumber\\
    &=\left\{\bu\in 4^{[n]} \,|\, C^\dagger U^\dagger \sigma_\bu U C\in  \pm \{I,Z\}^{\otimes n}\right\}.
\end{align} 
and
\begin{align}
    &C\{I,Z\}^{\otimes n} C^\dagger \times \left( U^\dagger \cP_n U\cap \cP_n\right) =\cP_n . \label{eq:noless}
\end{align}
The equality in equation (\ref{eq:noless}) is with respect to matrix-product modulo phases.
\end{theorem}

In the following, we prove several lemmas which will be used to show the unitary-stabilizer nullity and the state-stabilizer nullity coincide for the diagonal unitary.
\begin{lemma}\label{lem:pos}
For any $\bu\in 4^{[n]}$,
\begin{align}
    \sum_{\bx,\by\in \{0,1\}^n} |\langle \bx \sigma_\bu \by\rangle |\, |\bx\rangle \langle \by| \in \{I,X\}^{\otimes n}.\end{align}
    Here $|\langle \bx \sigma_\bu \by\rangle |$ is the norm of the complex number $\langle \bx |\sigma_\bu| \by\rangle$.
    \end{lemma}
\begin{proof}
Firstly we prove this lemma for $n=1$.  We can directly verify that
\begin{align}
 \sum_{\bx,\by\in \{0,1\}} |\langle \bx \sigma \by\rangle |\, |\bx\rangle \langle \by| =\left\{
 \begin{aligned}
 &I, \quad \text{If }\sigma \in\{I, Z\}\\
 & X, \quad \text{If }\sigma \in\{X, Y\}\\
 \end{aligned}
 \right.
\end{align}
It is worth noting that $\sum_{\bx,\by\in \{0,1\}} |\langle \bx \sigma \by\rangle |\, |\bx\rangle \langle \by|$ records the positions of the nonzero values in $\sigma$. 
From this point of view, we can generalize this proof to all $n$.
\end{proof}

By Theorem \ref{lem:compa} we prove that the lower bound for the $T$-count obtained by the unitary-stabilizer nullity will never perform worse than the bound obtained by the state-stabilizer nullity.
 In particular, we prove that for the diagonal unitary, the two bounds coincide.
 
We would use Eq.~\eqref{eq:diageqc} and set $C=H^{\otimes n}$. Before showing the main result of this section, we would like to present the following lemmas with detailed proof in Appendix D.
\begin{lemma}\label{lem:diagplus}
 For any diagonal $n$-qubit unitary $U$,
 \begin{align}
    &\{\bu\in 4^{[n]} \,|\, \langle +|^{\otimes n} U^\dagger \sigma_\bu U |+\rangle^{\otimes n}=\pm 1 \}\nonumber\\
    &=\{\bu\in 4^{[n]} \,|\, U^\dagger \sigma_\bu U \in \pm \{I,X\}^{\otimes n}\}.
\end{align}
\end{lemma}

\begin{lemma}\label{lem:prodH}
For any diagonal unitary $U$,
\begin{align}
     & \{I,X\}^{\otimes n}  \times \left( U^\dagger \cP_n U\cap \cP_n\right) =\cP_n.
\end{align}
 w.r.t product modulo phases. 
\end{lemma}

Now we are ready to establish the result for diagonal unitary as follows.
\begin{theorem}\label{theo:diag}
For any diagonal $n$-qubit unitary $U$, 
\begin{align}
    v(U)= \max_{C\in \cC_n} v_s(U C\ket{0}^{\otimes n}).
\end{align}
\end{theorem}
\begin{proof}
By setting $C=H^{\otimes n}$ in Theorem \ref{lem:compa} and then using Lemma \ref{lem:diagplus}, Lemma \ref{lem:prodH} to show that equations $(\ref{eq:diageqc})(\ref{eq:noless})$ are satisfied.
\end{proof}
%\begin{corollary}
%For any diagonal $n$-qubit unitary $U$, $n$-qubit Clifford $C\in \cC_n$,
%\begin{align}
 %   v(CUC^\dagger )= \max_{C\in \cC_n} v_s(U C\ket{0}^{\otimes n})).
%\end{align}
%\end{corollary}
%\begin{proof}
%Since both the stabilizer nullity of unitary and state are invariant under Clifford unitary.
%\end{proof}

%Especially,  Theorem \ref{theo:diag}  implies that 
%for diagonal unitary $U$, the maximum of its corresponding state-stabilizer nullity, i.e $v_s(U\ket{\psi})$ for some stabilizer state $\ket{\psi}$, is achieved by the $\ket{\psi}=\ket{+}^{\otimes n}$.
\begin{corollary}
For any diagonal $n$-qubit unitary $U$, the maximum of its corresponding state-stabilizer nullity, i.e. $v_s(U\ket{\psi})$ over all $n$-qubit stabilizer state $\ket{\psi}$, is achieved by $\ket{\psi}=|+\rangle^{\otimes n}.$ 
\end{corollary}

\subsection{Explicit expressions of the unitary-stabilizer nullity for unitary families}\label{sec:expli}

As proved in the previous sections, for the task of lower bounding the $T$ count for a given unitary, the bound obtained by the unitary-stabilizer nullity is always no less than the bound obtained by the state-stabilizer nullity. 
In this section, we give an example where the unitary-stabilizer nullity performs strictly better.

In Proposition 4.1 of Ref.~\cite{beverland2020lower} it showed that for $n\geq 3$, the state-stabilizer nullity of the states related to the $n$-bit controlled $Z$ gate is $n$, i.e. $v_s(C^{n-1}Z\ket{+}^{\otimes n})=n$ . This implies  the $T$ count for $C^{n-1}Z$ is at least $n$. One may wonder whether the unitary-stabilizer nullity can improve \diff{these} bounds. 
Unfortunately, since  $C^{n-1}Z$ is diagonal, by Theorem \ref{theo:diag} we know its unitary-stabilizer nullity is also $n$. To find examples to differentiate the state- and unitary-stabilizer nullity, we should focus on nondiagonal unitaries.

%\section{Examples}
In the following,  we give a family of $n$-qubit unitary, which has unitary-stabilizer nullity $2n$. This implies constructing such a unitary needs at least $2n$ $T$ gates by Theorem \ref{thm:lbt}. This unitary, or its corresponding quantum circuit, consists \diff{of} an $n$-qubit controlled $Z$ gate, a layer of Hadamard gate, and  an $n$-qubit controlled $Z$ gate. 
\begin{theorem}\label{theo:2n}For $n\geq 3$,
\begin{align} 
    v(C^{n-1}Z\,H^{\otimes n}\, C^{n-1}Z)=2n.
\end{align}
\end{theorem}
Theorem \ref{theo:2n} implies that the lower bound given by the unitary-stabilizer nullity can outperform the bound given by the stabilizer nullity. The detailed proof is provided in Appendix C. We further have the following corollary.

\begin{corollary}\label{cor0:2n}
There exists a family of $n$-qubit unitary $U_n$ such that 
\begin{align}
    v(U_n)> \max_{C\in \cC_n} v_s(U_n C\ket{0}^{\otimes n})).
\end{align}
\end{corollary}

%%%%%%%%%%%%%%%%%%%%%%%%%%%%%%%%%%%%%%%%%%%%%%%%%%%%%%%%%%%%%%%%%%%%

\subsection{Adding auxiliary systems improves the state-stabilizer nullity}\label{sec:aux_better}

In this section, we reconnect the state-stabilizer nullity and the unitary-stabilizer nullity. For any $n$-qubit unitary $U$, instead of comparing $v_s(U\ket{\psi})$ and $v(U)$ for some $n$-qubit stabilizer state $\ket{\psi}$ as in section \ref{sec:compsu}, we  add auxiliary systems and carefully choose the stabilizer states in the computation of $v_s(\cdot)$, i.e. $v_s(I_{2^d}\otimes U |\phi\rangle)$ for some $d$-qubit auxiliary system and $(d+n)$-qubit stabilizer state $\ket{\phi}$.
By Theorem \ref{thm:com_aux} and Corollary \ref{cor:aux_better},
we prove that adding auxiliary systems can strictly increase the value of the state-stabilizer nullity, thus strictly improving  the lower bound for the $T$ count through  Theorem \ref{thm:state_lower}. Furthermore, we prove that the unitary-stabilizer nullity $v(U)$ equals the best lower bound for the $T$ count which one can get through this method, i.e. computing the state-stabilizer nullity with auxiliary systems, as in Theorem \ref{thm:no_more}.

It is worth noting that although the state-stabilizer nullity $v_s(\cdot)$ can derive a lower bound for the $T$ count,  it is not the $T$ count itself, and it holds different properties to the $T$ count. Thus although  the $T$ count for $I_{2^d}\otimes U$ is no more than the $T$ count for $U$, the state-stabilizer nullity for $I_{2^d}\otimes U$ (with respect to some stabilizer state) can be strictly larger than the state-stabilizer nullity for $U$, as in Corollary \ref{cor:aux_better}.

\begin{theorem}\label{thm:com_aux}
	Let $\ket{\Phi}$ be the maximally entangled states on $2n$-qubit systems, i.e. $\ket{\Phi}=\sum_{\bx\in\{0,1\}^n} \frac{1}{\sqrt{2^n}}\ket{\bx}\ket{\bx}$. For any $n$-qubit unitary $U$,
	\begin{align}
		v(U)=v_s\left( I_{2^n}\otimes U \ket{\Phi} \right).
	\end{align}
\end{theorem}
\begin{corollary}\label{cor:aux_better} For the $2n$-qubit maximally entangled states  $\ket{\Phi}=\sum_{\bx\in\{0,1\}^n} \frac{1}{\sqrt{2^n}}\ket{\bx}\ket{\bx}$, for any $n$-qubit unitary $U$,
\begin{align}
	v_s(I_{2^n}\otimes U \ket{\Phi})
\geq \max_{C\in \cC_n} v_s(UC\ket{0}^{\otimes n}).
\end{align}
The above inequality can be strict, that is 
\begin{align}
	\exists U \text{ such that } v_s(I_{2^n}\otimes U \ket{\Phi})
> \max_{C\in \cC_n} v_s(UC\ket{0}^{\otimes n})\label{eq:aux_strict}.
\end{align}
\end{corollary}

Before giving  proofs to the theorem and the corollary, we want to make several clarifications:
\begin{itemize}[leftmargin=*]
	\item One can easily check the following argument holds as in Appendix \ref{sec:non_incre}
	\begin{align}\label{eq:aux_trivial}
		\max_{C'\in \cC_{2n}} v_s \left(I_{2^n}\otimes U\right) C'\ket{0^{\otimes 2n}})\geq \max_{C\in \cC_n} v_s(UC\ket{0}^{\otimes n}).
	\end{align} 
		That is when adding the auxiliary systems, the best lower bounds for the $T$ count we can get by the state-stabilizer nullity, is never less than the case when there is no auxiliary system. Compared to this easy argument,  Theorem \ref{thm:com_aux} has two key differences. The first one is finding the best bound in the left side of Eq.~(\ref{eq:aux_trivial}), or equivalent enumerating all possible $2n$-qubit stabilizer state $C'\ket{0^{2n}}$, need at least %$O(2^{(1/2+o(1))n^2})$ 
		$O(2^{\Omega(n^2)})$ 
		time and thus is hard. In contrast, in Corollary \ref{cor:aux_better} and Theorem \ref{thm:no_more} we show that to find the best bound,  it suffices to set the $2n$-qubit stabilizer state to be a fixed state, that is the maximally entangled state. The second difference is Eq.~(\ref{eq:aux_trivial}) shows only $\geq$, Corollary \ref{cor:aux_better} shows that this inequality can be strict $>$ in some cases. Combined with Theorem \ref{thm:state_lower},  this implies analyzing the state-stabilizer nullity with auxiliary systems can strictly improve the lower bound for the $T$ count.
		\item One may expect whether using higher-dimensional auxiliary systems or choosing other stabilizer states other than the maximally entangled states can further improve the bound by the state-stabilizer nullity. The answer is no. See as in Theorem \ref{thm:no_more}.
	\item Since in general the state-stabilizer nullity does not satisfy the subadditivity under composition (counterexample can be seen as in Appendix \ref{sec:not_subadd}), thus using properties of the state-stabilizer nullity and Theorem \ref{thm:com_aux} is not sufficient to imply the subadditivity under composition of the unitary-stabilizer nullity, that is Theorem \ref{thm:comp}.
		\item  Corollary \ref{cor:aux_better} and equation (\ref{eq:aux_strict}) do not violate the property that the state-stabilizer nullity is invariant under Clifford unitaries, i.e. for any unitary $V$, one \diff{has} $v_s(C V|\psi\rangle)=v_s(V|\psi\rangle)$ for any Clifford unitary $C$ and stabilizer state $\ket{\psi}$. More clarifications can be seen at Appendix \ref{sec:not_vio_inv}.
	\end{itemize}

\begin{proof}[of Theorem \ref{thm:com_aux}] 
By the symmetry of the maximally entangled states,
	for any $n$-qubit operator $M\in \mathbb{C}^{2^n\otimes 2^n}$, we have $(M\otimes I_{2^n})\ket{\Phi}=(I_{2^n}\otimes M^T) \ket{\Phi}$, which is known as the transpose trick. 
	
	For simplicity, in the following we abbreviate $I_{2^n}$ as $I$.
	 Suppose 
	the $2n$-qubit maximally entangled state  $\ket{\Phi}$ is on systems $A\otimes B$, where both $A,B$ are  $n$-qubit systems. Then for any $2n$-qubit Pauli operator $\sigma_\bu\otimes \sigma_\bv$, where $\bu,\bv\in 4^{[n]}$, the corresponding state Pauli function satisfies
	\vspace{0.5em}
	\begin{align}
		&P_{(I\otimes U)\ket{\Phi}}(\bu\bv)\nonumber\\
		&=\tr_{AB}\left(\sigma_\bu\otimes \sigma_\bv \, (I\otimes U) \,\ket{\Phi}\bra{\Phi} \,(I\otimes U^\dagger)\right) \nonumber\\
		&= \tr_{AB}\left( (I\otimes U^\dagger)\, \sigma_\bu\otimes \sigma_\bv \, (I\otimes U) \,\ket{\Phi}\bra{\Phi}
		\right) 
		\nonumber\\
% 		&=\tr_{AB}\left(\,(\sigma_\bu\otimes U^\dagger \sigma_\bv U)\, \ket{\Phi}\bra{\Phi}\right) &\\
		&= \tr_{AB}\left(\,(I\otimes U^\dagger \sigma_\bv U)\,(\sigma_\bu \otimes I)\, \ket{\Phi}\bra{\Phi}\right) &\nonumber\\
		&= \tr_{AB}\left(\,(I\otimes U^\dagger \sigma_\bv U \sigma_{\bu}^T)\ket{\Phi}\bra{\Phi}\right)\nonumber\\
		&= \gamma \tr_{AB}\left(\,(I\otimes U^\dagger \sigma_\bv U \sigma_{\bu})\ket{\Phi}\bra{\Phi} \right).
	\end{align}
	The last two are obtained by the transpose trick and the fact that $\gamma\in\{-1,+1\})$.
	Further, notice that
	\begin{align}
		\ket{\Phi}\bra{\Phi}&=\frac{1}{2^n}\sum_{\bx,\by\in\{0,1\}^n} \ket{\bx}\ket{\bx}\bra{\by}\bra{\by}.
	\end{align}
	The above equations can be further simplified by
	\begin{align}
		&P_{(I\otimes U)\ket{\Phi}}(\bu\bv)\nonumber\\
		&= \gamma \tr_{AB}\left(\,(I\otimes U^\dagger \sigma_\bv U \sigma_{\bu})\frac{1}{2^n}\sum_{\bx,\by\in\{0,1\}^n} \ket{\bx}\ket{\bx}\bra{\by}\bra{\by}\right) &\nonumber\\
		&= \gamma \tr_{B}\left( \, U^\dagger \sigma_\bv U \sigma_{\bu} \frac{1}{2^n}\sum_{\bx\in\{0,1\}^n} \ket{\bx}\bra{\bx}\right) \nonumber\\
		&= \gamma \tr_{B}\left( \, U^\dagger \sigma_\bv U \sigma_{\bu}\frac{1}{2^n}I\right) \nonumber\\
		&=\frac{\gamma }{2^n} \tr_{B}\left(  \sigma_\bv U \sigma_{\bu}U^\dagger\right) &\nonumber\\
		&={\gamma }\, P_U(\bv|\bu).
	\end{align}

In other words, if we ignore the $\gamma=\pm 1$ factor, then there is a one-to-one correspondence between the values of state Pauli functions  and the unitary Pauli function. Thus the two Pauli functions have the same number of $\pm 1$s, that is
$s(U)= s((I\otimes U)\ket{\Phi})$.
Since $\ket{\Phi}$ is a $2n$-qubit state, we obtain $2n - s(U)= 2n- s((I\otimes U)\ket{\Phi})$.
Then we have $v(U)=v_s\left( (I\otimes U) \ket{\Phi} \right)$.
\end{proof}

\vspace{0.5em}
Moreover, the proof of Corollary \ref{cor:aux_better} can be proved by combining Theorem \ref{thm:com_aux}, Theorem \ref{lem:compa} and Corollary \ref{cor0:2n}.

In the following, we find that $v(U)$ equals the best lower bound for the $T$ count that can be obtained from the state-stabilizer nullity with auxiliary systems. In the following theorem the maximization is over $d$-qubit auxiliary systems for any positive integer $d$. The detailed proof is provided in Appendix D.

\begin{theorem}\label{thm:no_more}For any $n$-qubit unitary $U$, we have
    \begin{align}
        v(U) =\max_{d\in \mathbb{N}^+} \max_{C\in\cC_{d+n}} v_s\left( (I_{2^d}\otimes U)\, C\ket{0}^{\otimes (d+n)} \right).
    \end{align}
\end{theorem}

\subsection{Estimating the $T$ count for quantum gates of interests}\label{subsec: examples}
In this section, we showcase the specific applications of unitary-stabilizer nullity in gate synthesis.
We present numerical results of using 
 the state-stabilizer nullity, the unitary-stabilizer nullity, and the stabilizer extent~\cite{bravyi2019simulation} to derive lower bounds for the $T$ count. We compare their performance and summarize the results in Table \ref{table:lowerbound}. 
  The definition and properties of stabilizer extent \diff{are} put in Appendix \ref{appendix:stabext}.
The quantum Fourier transform QFT$_n$ is the unitary where  QFT$_n |j\rangle := \frac{1}{\sqrt{2^n}} \sum_{k=0}^{2^n-1} e^{2\pi i jk/2^n}\ket{k}$.

As shown in table \ref{table:lowerbound}, when compared with state-stabilizer nullity, 
the lower bounds obtained by unitary-stabilizer nullity 
is the same as for diagonal gates, and better in the general case.
When compared with stabilizer extent, for small $n$, the unitary-stabilizer nullity  outperforms for some gates. For larger $n$, the stabilizer extent might perform better, but it is harder to compute  since the cardinality of stabilizer states increases as $2^{\Omega(n^2)}$~\cite{gross2006hudson}.

\begin{table*}[ht] 
    \centering
    \begin{tabular}{c|c|c|c|c}
          Unitary  & Corresponding state & \multicolumn{3}{|c}{\makebox[0pt]Lower bounds obtained by}  \\
          & & state-stabilizer nullity & unitary-stabilizer nullity & stabilizer extent\\
          \hline
          $T$    &  $T\ket{+}$ &  1  &  1  & 1\\
          $\sqrt{T}$ & $\sqrt{T}\ket{+}$ &  1  &  1  &0.7549\\
          $CCZ$ & $CCZ \ket{+}^{\otimes 3}$ &  3  &  3  & 3.6349\\
          $C^3Z$ & $C^3Z \ket{+}^{\otimes 4}$ & 4 & 4 & 5.1226\\
          $\sqrt{T}H\sqrt{T}$ & $\sqrt{T}H\sqrt{T}\ket{+}$ &  1   & 2 &0.9109\\
          $THT$ & $THT \ket{+}$ &  1   &  2  & 1.4542\\
          QFT$_3$ & QFT$_3\ket{x}$ &     &   & \\
         &  $x \equiv 0 \mod 2 $ &  0   & 4  & 0.0000\\
         &  $x \equiv 1 \mod 2 $ &  1   & 4  & 1.0000\\
         QFT$_4$ & QFT$_4\ket{x}$ &     &   & \\
           & $x \equiv 1 \mod 2 $ &  2  &  6  & 1.7555\\ 
             & $x \equiv 0 \mod 4 $ &  0  &  6  & 0.0000\\
               & $x \equiv 2 \mod 4 $ &  1  &  6  & 1.0000\\
    \end{tabular}
	\caption{Lower bounds for $T$ count for the given unitary. The first column is the given unitary $U$. The second column is the corresponding state $\ket{\psi}$. The next three columns are the lower bounds obtained by state-stabilizer nullity $v_s$, the unitary-stabilizer nullity $v$, and the stabilizer extent $\xi$. The corresponding value are $v_s(\ket{\psi})$, $v(U)$ and $\frac{\log(\xi(\ket
	{\psi})}{\log \xi(\ket{T})}$,  as proved in Theorem~\ref{thm:lbt}, Theorem~\ref{thm:state_lower}. Since QFT$_n\ket{+}^{\otimes n}=\ket{0}$,  similar to Ref.~\cite{beverland2020lower} instead we use QFT$_n\ket{x}$ as its corresponding states.}   \label{table:lowerbound}
\end{table*}

\section{Conclusion}\label{sec:open}

\diff{Resource theory for magic monotones has wide applications in quantum information. In this work, we introduce a magic monotone for unitaries, and use it to lower bound the $T$ count in the circuit synthesis task. This lower bound quantifies the minimum magic resources required for the computation,  and lower bounds the runtime of classical simulation algorithms based on the Clifford+$T$ gate set. We also give  numerical results which show that our magic monotone can give better $T$ count lower bounds for some gates when compared with other magic monotones.}

\diff{Specifically,} in this paper, for any $n$-qubit unitary $U$, we introduce its unitary-stabilizer nullity as $v(U)=2n-s(U)$, where $s(U)$ is the number of $\pm 1s$ in the Pauli transfer matrix. We prove that the unitary-stabilizer nullity satisfies several desirable properties, and in particular establish that it is a lower bound for the $T$ count of $U$. We show that the lower bound obtained by the unitary-stabilizer nullity never performs worse than the state-stabilizer nullity. We also give an $n$-qubit unitary family of unitary-stabilizer nullity $2n$, providing examples  where the unitary-stabilizer nullity strictly outperforms.  For the diagonal unitaries, we prove the lower bounds obtained by the state- and the unitary-stabilizer nullity are the same. Besides, we connect the unitary-stabilizer nullity and the state-stabilizer nullity with auxiliary systems. Moreover, the results in this work may shed lights for the study of resource theories of quantum dynamics (e.g., Refs.~\cite{Diaz2018,Wang2019,Rosset2017a,Li2018c,Regula2021,Liu2019b,Wang2018g}).

\diff{Meaningful magic monotones with desirable properties are very useful to exploit the power of nonstabilizer resources in 
quantum computation. Our work has, in particular, shown the applications of unitary-stabilizer nullity in estimating the $T$ count, which is central for resource estimation for fault-tolerant computation~\cite{Bravyi_2005}, performance analysis for classical-simulation algorithms~\cite{seddon2021quantifying,bravyi2016improved,bravyi2016trading},   
and even for the complexity of the 
parameterized Quantum Merlin-Arthur problem
(QMA problem)~\cite{arunachalam2022parameterized}.
There are many other applications beyond. For one thing, we could apply the unitary-stabilizer nullity to benchmark the gate synthesis task under stabilizer operations~\cite{welch2014efficient,heyfron2018efficient,ross2016optimal}. Using the result in Theorem 13, the benchmark method can be generalized for any non-Clifford gates. For another, one might use the unitary-stabilizer nullity to quantify the nonstabilizerness of quantum circuits~\cite{oliviero2022measuring,haug2022scalable} on near-term quantum devices, as we could extract the unitary Pauli function on the quantum hardware, via firstly measuring the states $U(\sigma_v+I)/{2^n}U^\dagger$ and $U(I-\sigma_\bv)/{2^n}U^\dagger$ with measurement operator $\sigma_\bu$, and then obtain the Pauli function.
% by the fact that $\sigma_\bv=\frac{1}{2}\left[(\sigma_\bv+I)-(I-\sigma_\bv)\right]$ and the linearity of trace operations. 
%i.e. $P_{U}(\bu|\bv) = \frac{1}{2} \tr (\sigma_{\bu}U(\sigma_v+I)/{2^n}U^\dagger)-\frac{1}{2} \tr (\sigma_{\bu}U(-\sigma_v+I)/{2^n}U^\dagger)$.
Protocols for evaluating the nonstabilizerness experimentally~\cite{oliviero2022measuring,haug2022scalable} can be used to analyze the performance of quantum devices.}

Beyond this work, there are still various interesting open questions 
on the quantification of non-Clifford resources:  
\begin{itemize}
 \item For the task of lower bounding $T$ gates for a given unitary, 
 the unitary-stabilizer nullity gives a bound up to $2n$. However, in general the $T$ count can be much larger than $2n$. For example, it is known that only unitaries with elements in certain rings can be exactly synthesized by Clifford and $T$ gates~\cite{giles2013exact}. Is it possible to get a superlinear bound by devising alternative magic monotones? 
    \item Here we define the $T$ count to be the minimum number of $T$ gates used for exactly synthesizing a given unitary $U$. Is it possible to lower bound the number of $T$ gates for unitary $V$ which approximates $U$, i.e. $|V-U|\leq \epsilon$? Beverland \emph{et al.}~\cite{beverland2020lower} have considered this setting for single-qubit unitaries.
\end{itemize}

\textbf{Acknowledgements.} 
We thank the anonymous reviewers for their valuable suggestions, which help us better present the results and proofs. This work was done when J. J. was a research intern at Baidu Research.

% \bibliographystyle{ieeetr}
% \bibliographystyle{IEEEtran}
%  \bibliographystyle{apsrev4-1}

%\clearpage

\vspace{2cm}
\onecolumngrid
\vspace{2cm}
% \begin{center}
% {\textbf{\large Supplemental Material}}
% \end{center}

\appendix

\section{Equivalence between the definitions of the state-stabilizer nullity}\label{appendix:equ_state_stab}

In this section, we show Definition \ref{def:ssn}  is equivalent  to the definition of the state-stabilizer nullity in the original paper~\cite{beverland2020lower}.
Recall that  for an $n$-qubit state $\ket{\psi}$,
Beverland \emph{et al}~\cite{beverland2020lower} define the state-stabilizer nullity as
\begin{align}
	\nu(\ket{\psi})=n-\log_2{|\mathrm{Stab}\, \ket{\psi}}|,
\end{align}
 where $\mathrm{Stab}\, \ket{\psi}$ is the stabilizer of $\ket{\psi}$, that is
 \begin{align}
 	\mathrm{Stab}\, \ket{\psi}=\{P\in \hat{\cP}_n: P|\psi\rangle=|\psi\rangle\}.
 \end{align}
 Here we prove our definition $v_s(\ket{\psi})$ is equivalent to $\nu(\ket{\psi})$, that is $v_s(\ket{\psi})=\nu(\ket{\psi})$, by building a one-to-one correspondence between $\bu\in 4^{[n]}$ where 
$P_{\ket{\psi}}(\bu)=\pm 1$, and $P\in\hat{\cP}_n$ where $P\ket{\psi}=\ket{\psi}$. The key is noticing that the state Pauli functions are ranging over the quotient group $\cP_n$, while the $\mathrm{Stab}$ is ranging over the Pauli group $\hat{\cP}_n$ where
\begin{align}
	&\cP_n=\{\sigma_\bu|\bu\in 4^{[n]}\}\label{eq:A3},\\
	&\hat{\cP}_n=\{\pm \sigma_\bu, \pm i \sigma_\bu|\bu\in 4^{[n]}\} \label{eq:A4}.
\end{align}
More precisely,

$\rightarrow$ For one direction, if any $\bu$ satisfies that 
$P_{\ket{\psi}}(\bu)=\gamma$ where $\gamma=\pm 1$, since $\sigma_\bu$ only has $\pm1$ eigenvalues and $\ket{\psi}$ is unit, we must have 
\begin{align}
	\sigma_\bu \ket{\psi} =\gamma\ket{\psi}.
\end{align}
Thus we find $P^{(\bu)}:=\gamma \sigma_\bu \in \hat{\cP}_n$ such that
\begin{align}
	P^{(\bu)}\in \mathrm{Stab}\, \ket{\psi}.
\end{align}
Further notice that for $\bu\neq\bv$ where $P_{\ket{\psi}}(\bu)=P_{\ket{\psi}}(\bv)=\pm 1$,  we have $P^{(\bu)} \neq P^{(\bv)}$. Thus the map from $\bu$ to $P^{(\bu)}$ is \diff{a bijection}.

$\leftarrow$ For the other direction,  notice that  
since $i\sigma_\bu$ has eigenvalues $\pm i$, we have 
\begin{align}
	 i\sigma_\bu \ket{\psi}\neq \ket{\psi}, \forall \ket{\psi}.
\end{align} Then for any $P\in \mathrm{Stab}\,\ket{\psi}$, there exists $\bu_P\in 4^{[n]}$ such that $P=\gamma\sigma_{\bu_P}$ where $\gamma=\pm 1$. We can verify that 
\begin{align}
	P_{\ket{\psi}}(\bu_p)=\gamma.
\end{align}
Also notice that by the linearity, if $P\in\mathrm{Stab}\,\ket{\psi}$, then $-P\not\in\mathrm{Stab}\,\ket{\psi}$. Thus 
 for any $P\neq P'\in\mathrm{Stab}\,\ket{\psi} $, $\bu_P\neq \bu_{P'}$. Thus the map from $P\in \mathrm{Stab}\,\ket{\psi}$ to $\bu_P$ is also a bijection.

\diff{Combing the above argument that there is a bijection from $\bu$ where $P_{\ket{\psi}}(\bu)=\pm 1$ to the elements of $\mathrm{Stab}\,\ket{\psi}$ and a bijection vice versa,
we conclude that 
\begin{align}v_s(\ket{\psi})=\nu(\ket{\psi}).
\end{align}
}
\section{Non-Clifford resources in two computational models}\label{appendix:relation}

In this section, we specify two  computation models which we mentioned  in this manuscript, that is the quantum computation model with magic states, and the quantum computation model with unitary. 
We would use examples from~\cite{howard2017application,jones2013low,gosset2014algorithm} to  clarify the relationship between the number of non-Clifford resources used in the two models, that is  Proposition \ref{pro:T2Ts} and Proposition \ref{pro:Tsn2T}.

%%% Two computational model
To begin with, we firstly clarify the two computation models. 
 We write $T$ gate as $T=\begin{bmatrix}1&0\\0& e^{i\pi/4}\end{bmatrix}$, T state as $\ket{T}=T\ket{+}.$

\subsection{Two computation models}

\noindent\textbf{Quantum computation with magic state~\cite{Bravyi_2005}.}
In this model, the computational task is to construct a target $n$-qubit quantum state $\rho$, either a pure state or a mixed state. 
The valid operations are as follows.

\begin{itemize}
    \item Initiate qubits in state $\ket{0}^{\otimes m}\ket{T}^{\otimes s}$ for some integer $m, s$.
    \item Perform stabilizer operations,  including using Clifford unitaries and qubit measurement in the computational basis. Here the Clifford unitary is performed adaptively, that is the Clifford unitary used  can depend on previous measurement outcomes.
    \item After the computation, use the first $n$ qubits as the output, which should be in state $\rho$.
\end{itemize}

There are variants of this model which allow using catalyst states in the initiation, we do not consider those variants here. For readers who are interested, more can be seen in \cite{beverland2020lower}.

\vspace{1em}

\noindent\textbf{Quantum computation with unitary.} In this model, the computational task is to construct a target $n$-qubit unitary $U$ using basic gates. More specifically, the task is to decompose $U$ into a sequence of gates from the Clifford+$T$ gate set, without using ancillary qubits and measurement. Mathematically, we aim to find unitary $U_1,...,U_l$ such that
\begin{align}\label{eq:T}
    & U= U_1U_2...U_l,
\end{align}
where each $U_i$ corresponds to either a  Clifford gate, or a $T$ gate. Here $l$ is an integer corresponding to the size of the circuit. 

\vspace{0.7em}

\noindent\textbf{Differences:} Note that in  the quantum computation with magic state, one can use auxiliary qubits, do measurements and perform Clifford unitary adaptively, while in the quantum computation with unitary, one cannot use auxiliary qubits, cannot do measurements and there are no adaptive operations. 
\subsection{Non-Clifford resources used in the two models}
In the quantum computation with magic state model, the T state $\ket{T}$ is the non-Clifford resource, while the stabilizer operations \diff{are} viewed as free operations. For an $n$-qubit state $\rho$,
 we use $w_s(\rho)$ to denote the minimum $s$, that is the minimum number of $\ket{T}$ needed to  construct $\rho$. We denote $w_s(\rho)=+\infty$ if $\rho$ cannot be exactly constructed through the above process. 
 
In the quantum computation with unitary model, the $T$ gate is the non-Clifford resources, while the Clifford \diff{gates}, or Clifford \diff{unitaries}, are viewed as free resources. 
We use  $t(U)$ to denote the minimum number of $T$ gates used in Eq.~(\ref{eq:T}). We denote $t(U)=+\infty$ if $U$ cannot be exactly synthesized with Clifford+$T$ gate set. We also call $t(U)$ the $T$ count of $U$.

\subsection{Relation between number of non-Clifford resources in two models}

As we inexplicitly used in the manuscript, firstly the lower bound of $\ket{T}$ can derive a lower bound of $T$. This fact is used in \cite{beverland2020lower}.

\begin{proposition}\cite{beverland2020lower}\label{pro:T2Ts}
For an $n$-qubit unitary $U$, for any $n$-qubit stabilizer state $\ket{\phi}$,
\begin{align}
    t(U)\geq w_s(U\ket{\phi}).
\end{align}
\end{proposition}
\begin{proof}
This proof is direct. Suppose we  decompose $U$ into sequences of gates in Clifford+$T$ gate set with in total $t(U)$ $T$ gates. Then use gate injection~\cite{zhou2000methodology,fowler2012surface} we can construct state $U\ket{\phi}$ by using $t(U)$ copies of $\ket{T}$.
\end{proof}

\diff{Furthermore, in} the other direction, 
one may wonder if $t(U)$ can be completely characterized by $w_s(\cdot)$, via enumerating all stabilizer states $\ket{\phi}$. We emphasize here this is not true. 
\begin{proposition}\label{pro:Tsn2T}~\cite{howard2017application,jones2013low,gosset2014algorithm}
There exists $3$-qubit unitary $U$ such that 
\begin{align}
    t(U)> \max_{\ket{\phi}} w_s(U\ket{\phi}),
\end{align}
where the maximum is over all $3$-qubit stabilizer states $\ket{\phi}$.
\end{proposition}
\begin{proof}
This fact is implied by~\cite{howard2017application,jones2013low,gosset2014algorithm}. 
Here $U$ is set to be the Toffoli gate. 
Using algorithm which outputs exactly $w(\cdot)$, Gosset shows $t(U)=7$~\cite{gosset2014algorithm}. On the other hand, by utilizing the adaptive computation, 
$w_s(U\ket{\phi})\leq 4$~\cite{jones2013low}.
\end{proof}

\subsection{Background of stabilizer extent}
\label{appendix:stabext}
For completeness, here we write down the definition and properties of stabilizer extent, which is from \cite{beverland2020lower}.

\begin{definition}[\cite{beverland2020lower}.]\label{def:stabext}
	For an arbitrary pure state $\ket{\psi}$, the stabilizer extent $\xi(\ket{\psi})$ is defined as
	\begin{align}
	\xi(\ket{\psi}) = \min \|(c_1,c_2,...,c_k)\|_1^2 \text{ s.t } \ket{\psi}=\sum_{\alpha=1}^k c_\alpha \ket{\phi_\alpha},
	\end{align}
	where the minimization is over all complex linear combination of stabilizer states $\{\ket{\phi_{\alpha}}\}$.
\end{definition}

\begin{theorem}[\cite{beverland2020lower}] \label{thm:lbstabext}
For any $n$-qubit unitary $U$, for any single-qubit stabilizer states $\ket{\psi_1},...,\ket{\psi_n}$,  we have
\begin{align}
 t(U) \geq 
 \frac{   \log\xi(
           U \ket{\psi_1}...\ket{\psi_n}
           )
        }
        {\log \xi( \ket{T})}. 
\end{align}
\end{theorem}

\section{Clarification on properties of the state-stabilizer nullity}\label{sec:not_vio_inv}
\subsection{Corollary \ref{cor:aux_better} does not violate properties of the state-stabilizer nullity}

Corollary \ref{cor:aux_better} and equation (\ref{eq:aux_strict}) does not violate the property that the state-stabilizer nullity is invariant under Clifford unitaries, i.e., for any $n$-qubit unitary $V$, for any  $n$-qubit Clifford unitary $C\in\cC_n$ and any $n$-qubit stabilizer state $\ket{\psi}$, one have 
\begin{align}\label{eq:C1}
	v_s(C V|\psi\rangle)=v_s(V|\psi\rangle).
\end{align}
  What the corollary implies is slightly different. Specifically, what it implies is $\exists V, C$ such that
 \begin{align}
 	v_s(VC\ket{\psi})\neq v_s(V\ket{\psi}),
 \end{align}
 Or equivalently there may exist different $n$-qubit  stabilizer states $\ket{\psi_1}$ and $\ket{\psi_2}$ such that
 \begin{align}\label{eq:C2}
 	v_s(V\ket{\psi_1})\neq v_s(V\ket{\psi_2}).
 \end{align}
 
 Write $\ket{\Phi}$ to be the $2n$-qubit maximally entangled states,
Corollary \ref{cor:aux_better} implies $\exists$ $n$-qubit unitary $U$ such that setting $V=(I_{2^n}\otimes U)$
\begin{align}
	v_s(V\ket{\Phi})> v_s(V(I_{2^n}\otimes C'')\ket{0^{\otimes 2n}}),
\end{align}
for any $n$-qubit Clifford $C''\in \cC_n$.

\subsection{The state-stabilizer nullity does not satisfy the subadditivity under composition}\label{sec:not_subadd} 
In this section, we show that for any $1$-qubit stabilizer state $\ket{\psi}$, there exist single-qubit unitaries $U,V$ such that 
\begin{align}\label{eq:C5}
	v_s(UV\ket{\psi})> v_s(U\ket{\psi})+v_s(V\ket{\psi}).
\end{align}
Thus the state-stabilizer nullity does not satisfy the subadditivity under composition. 
 In particular, suppose 
\begin{align}
	\ket{\psi} = C\ket{0}, C\in \cC_1.
\end{align}
Set 
\begin{align}
 U=e^{iX} H C^{-1}, V= CHC^{-1}.
\end{align}
Since $\ket{+}$ is an eigenstate of $e^{iX}$, one can verify that  
\begin{align}
&v_s(U\ket{\psi})=v_s(e^{iX}\ket{+})=0.	
\end{align}
Since $CH$ is Clifford, then 
\begin{align}
&v_s(V\ket{\psi})=v_s(CH\ket{0})=0.
\end{align}
Finally, one can verify by calculating the state Pauli function \diff{that}
\begin{align}
v_s(UV\ket{\psi})=v_s(e^{iX}\ket{0})=1.	
\end{align}
Thus Eq.~(\ref{eq:C5}) holds.
This counterexample can be generalized to $n$-qubit case.

\subsection{Non-decreasing of state-stabilizer nullity with higher-dimensional systems}\label{sec:non_incre}
\begin{lemma} For any $n$-qubit unitary $U$,
	for any $d\geq d'$, we have 
	\begin{align}
		\max_{C\in \cC_{d+n}} v_s\left( (I_{2^d}\otimes U) C\ket{0}^{\otimes d+n}\right)\geq \max_{C'\in \cC_{d'+n}} v_s\left( (I_{2^{d'}}\otimes U) C'\ket{0}^{\otimes d'+n}\right).
	\end{align}
\end{lemma}
\begin{proof}
	\diff{It suffices to notice that, for any $C'\in \cC_{d'+n}$}, let
	\begin{align}
		C:=I_{2^{d-d'}}\otimes C'.
	\end{align}
	we have 
	\begin{align}
		v_s\left( 
		(I_{2^d}\otimes U) C\ket{0}^{\otimes d+n}\right)
		&=v_s\left( 
		(I_{2^{d-d'}}\otimes I_{2^{d'}}\otimes U) C\ket{0}^{\otimes d+n}\right) \nonumber\\
				&=v_s\left(\ket{0}^{d-d'}\otimes (I_{2^{d'}}\otimes U)C'\ket{0}^{\otimes d'+n}\right) \nonumber\\
				&=v_s\left(\ket{0}^{d-d'}\right)
+ v_s\left( (I_{2^{d'}}\otimes U)C'\ket{0}^{\otimes d'+n}\right)\label{eq:C17}\\
&=v_s\left( (I_{2^{d'}}\otimes U)C'\ket{0}^{\otimes d'+n}\right)\label{eq:C18}.
	\end{align} 
	\diff{The  Eq.~(\ref{eq:C17}) is obtained by the fact that the state-stabilizer nullity is  additive under tensor product. The Eq.~(\ref{eq:C18}) is by the fact that $v_s(\ket{\phi})=0$ for stabilizer state $\ket{\phi}$.} The proof of the two \diff{facts} can be seen in \cite{beverland2020lower}.
\end{proof}

\subsection{Detailed proofs of Lemma \ref{lem:Clifford10}, Lemma \ref{lem:10s}, and Lemma \ref{lem:pj}}\label{appendix:proof_Lemma_12}

 \begin{proof} [of Lemma \ref{lem:Clifford10}] We only give proof to the first statement, the second is similar. Recall that
 \begin{align}
    P_U(\bu|\bv) &= \tr(\sigma_\bu U\sigma_\bv U^\dagger )/2^n,\\
     \tr(\sigma_\bu \sigma_\bv) &= \left\{
     \begin{aligned}
     & 2^n, &\text{ if $\bu=\bv$}\\
     & 0, &\text{ else.}
     \end{aligned}
     \right.
 \end{align}
 
 Since $U$ is Clifford, $U(\cdot)U^\dagger$ is an isomorphism of $\hat{\cP}_n$, i.e.  $U \hat{\cP}_n U^\dagger= \hat{\cP}_n$. Thus $\forall \bu\in 4^{[n]}$, there exists  $\bv\in 4^{[n]}, k\in [4]$ such that $U i^k \sigma_\bv U^\dagger= \sigma_\bu,$ or equivalently $U  \sigma_\bv U^\dagger= (-i)^k \sigma_\bu $. Furthermore, 
 since the conjugate operation would not change the eigenvalues of $\sigma_\bv$, which \diff{are} $\pm 1$. Thus the global phase $(-i)^k$ must be $\pm 1$. Thus there exists at least one $\bv$ such that $P_{U}(\bu|\bv)=\pm 1$.
 
 On the other hand, by the linearity of $U(\cdot)U^\dagger$, if there exists $\bv\neq\bv'$ such that 
 \begin{align}
     U\sigma_\bv U^\dagger &= (-i)^{k_\bv} \sigma_\bu,\\
     U\sigma_{\bv'} U^\dagger &= (-i)^{k_{\bv'}} \sigma_\bu.
 \end{align}
 Then we have
     $U \left( (-i)^{k_{\bv'}} \sigma_\bv- (-i)^{k_{\bv}} \sigma_{\bv'}\right)  U^\dagger =0$,
 which leads to a contradiction by  the fact that $\sigma_\bv$ and $\sigma_{\bv'}$ are linear independent.
 Thus we prove that there \diff{is} only one such $\bv$ which satisfies $P_{U}(\bu|\bv)=\pm 1$. Since $U(\cdot)U^\dagger$ is an isomorphism of $\hat{\cP}_n$ by Eq.~(\ref{eq:c01}) we know for any $\bv'\neq \bv$, we have $P_U(\bu|\bv')=0$.
 \end{proof}
 
\vspace{3em}

\begin{proof}[of Lemma \ref{lem:10s}]
Since $\sigma_\bu$ is a Pauli operator,  $2^{n-1}$ number of its eigenvalues are $+1$, and $2^{n-1}$  eigenvalues are $-1$. We write
\begin{align}
    &\sigma_\bu=\sigma_\bu^+-\sigma_\bu^-,
\end{align}
where $\sigma_\bu^\pm$ are the projections onto the $\pm1$ eigenspace of $\sigma_\bu$. \diff{Similarly}, we write 
\begin{align}
    \sigma_\bv=\sigma_\bv^+-\sigma_\bv^-, \quad
    Q_\bv^+=U\sigma_\bv^+ U^\dagger, \quad
        Q_\bv^-=U\sigma_\bv^- U^\dagger.
\end{align}
 Then
\begin{align}
    P_U(\bu|\bv)&=\tr \sigma_\bu U \sigma_\bv U^\dagger/2^n\\
    &=\tr (2\sigma_\bu^+-I)(2Q_\bv^+-I)/2^n \text{ (Since $\sigma_\bu^+ +\sigma_\bu^-=I$)} \nonumber\\
    &=(4\tr \sigma_\bu^+Q_\bv^+-2^n)/2^n.
\end{align}
Since $P_U(\bu|\bv)=\pm 1$, we have either 
\begin{align}\label{equ:0}
    \tr \sigma_\bu^+Q_\bv^+=0.
\end{align}
or
\begin{align}\label{equ:1}
     \tr \sigma_\bu^+Q_\bv^+=2^{n-1}.
\end{align}
Note that the eigenvalues of $\sigma_\bu^+$ (or $Q_\bv^+$) are $2^{n-1}$ $1$-eigenvalue and $2^{n-1}$ $0$-eignenvalue. 

If equation (\ref{equ:0}) holds, we must have $Q_\bv^+=\sigma_\bu^-$. To prove this, let $\bx_1,...,\bx_{2^{n-1}}$ be an orthogonal basis of the $+1$-eigenspace of $Q_\bv$. Since $Q_\bv^+$ and $\sigma_\bu^-$ are all projections, it suffices to prove vector spaces $\{\bx|Q_\bv^+ \bx=\bx\}=\{\bx|\sigma_\bu^- \bx=\bx\}$. Since the two spaces have the same dimension of $2^{n-1}$, it suffices to prove 
\begin{align}
    \sigma_\bu^-\bx_i=\bx_i,\forall i.
\end{align}
or by the fact $\sigma_\bu^+ +\sigma_\bu^-=I$, equivalently to prove
\begin{align}
   \sigma_\bu^+\bx_i=0,\forall i.
\end{align}
If the above is not true, $\exists i_0,  \sigma_\bu^+\bx_{i_0}\neq 0$. Extend $\{\bx_i\}_{i=1}^{2^{n-1}}$ to the basis of the whole ($2^n$-dimensional) space, we have 
\begin{align}
    \tr \sigma_\bu^+Q_\bv^+&=\sum_{i=1}^{2^n} \langle \bx_i|\sigma_\bu^+Q_\bv^+|\bx_i\rangle
    =\sum_{i=1}^{2^{n-1}} \langle \bx_i|\sigma_\bu^+|\bx_i\rangle
    \geq \langle \bx_{i_0}|\sigma_\bu^+|\bx_{i_0}\rangle
    >0,
\end{align}
which leads to a contradiction. Thus we conclude $Q_\bv^+=\sigma_\bu^-$, that is 
\begin{align}
    U \sigma_\bv U^\dagger &=2Q_\bv^+- I
    = 2\sigma_\bu^- -I
    = -\sigma_\bu.
\end{align}
Then $\forall \bu'\neq \bu$, $P_U(\bu'|\bv)=0.$

If equation (\ref{equ:1}) holds, notice that
\begin{align}
    \tr \sigma_\bu^- Q_\bv^+&=\tr(I-\sigma_\bu^+) Q_\bv^+
    =2^{n-1}-2^{n-1}
    =0.
\end{align}
Similarly, we conclude $Q_\bv^+=\sigma_\bu^+$.
\end{proof}

\begin{proof}[of Lemma \ref{lem:congs}]
Notice that 
\begin{align}
    |\cP_U|&=| U\cP_n U^\dagger\cap \cP_n|= |U(\cP_n\cap U^\dagger \cP_n U)U^\dagger| \nonumber\\
    &= |\cP_n\cap U^\dagger \cP_n U|=|\cP_{U^\dagger}|.
\end{align}
In the following we show that $ |\cP_U|=s(U)$.
Since $\cP_U$ is an intersection of two groups, thus $\cP_U$ is a group. Since $\cP_U\subseteq\cP_n$, thus further $\cP_U$ is a subgroup of $\cP_n$. 
To prove $ |\cP_U|=s(U)$, it suffices to prove that there is an one-to-one correspondence between $\sigma_\bu \in \cP_U$ and a pair $(\bu,\bv)$ such that $P_U(\bu|\bv)=\pm 1$.

The  proof is as follows.
On one hand if $\sigma_\bu\in \cP_U = U\cP_n U^\dagger \cap \cP_n$,
then by definition there exists $\bv_\bu \in 4^{[n]}$ such that 
\begin{align}
	&U\sigma_{\bv_\bu} U^\dagger = \pm \sigma_\bu \label{eq:vu}.
	\end{align}
	Thus, we have  $P_U(\bu|\bv_\bu)=\pm 1$.
Further we know such $\bv_{\bu}$ is unique since  
Eq.~(\ref{eq:vu}) implies $U^\dagger \sigma_{\bu} U=\pm \sigma_{\bv_{\bu}}$, and for any $\bv\neq {\bv_{\bu}} $ we have $\tr(\sigma_{\bv_{\bu}}\sigma_\bv)=0$.

One the other hand, for any $(\bu,\bv)$ such that $P_U(\bu|\bv)=\pm 1$, by Lemma \ref{lem:10s} we have $U\sigma_\bv U^\dagger=\pm \sigma_\bu$ thus
\begin{align}
	\sigma_\bu\in \cP_U = U\cP_n U^\dagger \cap \cP_n.
\end{align}
\end{proof}

% \begin{lemma}\label{lem:pj}
% Suppose $A,B$ are subgroups of $H$, then
% \begin{align}
%     |A\times B|=\frac{|A|\cdot|B|}{|A\cap B|}
% \end{align}
% Since $A\times B\subseteq H$ we have
% \begin{align}
%   |A|\cdot |B|\leq |H|\cdot |A\cap B|. 
% \end{align}
% \end{lemma}
\begin{proof}[of Lemma \ref{lem:pj}]
For any $a_1,a_2\in A$, notice that
\begin{align}
    &a_2^{-1}a_1\in B \rightarrow a_1B=a_2B,\\
    & a_2^{-1}a_1\not\in B \rightarrow a_1B\cap a_2B=\emptyset.
\end{align}
Thus $A\times B$ can be written as the union of disjoint cosets: If we write $\hat{A}$ as the set of coset representatives, that is
a maximum set which satisfies 
\begin{align}
    \left\{
    \begin{aligned}
    & \hat{A}\subseteq A,\\
    & \forall a_1,a_2 \in \hat{A},\, a_2^{-1}a_1\not\in B.
    \end{aligned}
    \right.
   \label{eq:hatA}
   \end{align}
   Then we have
   \begin{align}
       & A\times B=\cup_{a\in\hat{A}} aB, \\
       & \forall a_1,a_2\in \hat{A}, a_1B\cap a_2B =\emptyset.
   \end{align}
   Thus 
   \begin{align}
       |A\times B|=|\hat{A}|\cdot|B|
       \label{eq:coset1}
   \end{align}
   Since $A,B$ are groups, thus $A\cap B$ is a group. Further notice that since $\hat{A}\subseteq A$, Eq.~\eqref{eq:hatA} is equivalent to
   \begin{align}
    \left\{
    \begin{aligned}
    & \hat{A}\subseteq A,\\
    & \forall a_1,a_2 \in \hat{A},\, a_2^{-1}a_1\not\in B\cap A.
    \end{aligned}
    \right.
   \end{align}
   This implies $\hat{A}$ is also the set of coset representative for $A\cap B$ in $A$. Thus 
   \begin{align}
       |\hat{A}|=\frac{|A|}{|A\cap B|}
       \label{eq:coset2}.
   \end{align}
   Combine Eq.~\eqref{eq:coset1} and Eq.~\eqref{eq:coset2}   we finally conclude that
   \begin{align}
         |A\times B|=\frac{|A|\cdot|B|}{|A\cap B|}.
   \end{align}
   \end{proof}

\section{Detailed proofs of main results}
\subsection{Detailed proof for Theorem~\ref{lem:compa}}
\begin{proof} [Proof of Theorem~\ref{lem:compa}]
Fix any  $n$-qubit stabilizer state, that is
    $\ket{\psi}=C\ket{0}^{\otimes n}$
 for some Clifford operator $C\in \cC_n.$
For any $\bu\in 4^{[n]}$, notice that the condition that $\ket{0}^{\otimes n}$ is an eigenvector of $C^\dagger U^\dagger \sigma_\bu U C$, would imply $ P_{U\ket{\psi}}(\bu)=\pm 1,$ then
\begin{align}\label{eq:ss}
   C^\dagger U^\dagger \sigma_\bu U C  \in  \pm \{I,Z\}^{\otimes n}\rightarrow  P_{U\ket{\psi}}(\bu)=\pm 1 .
\end{align}
  Note that in the following  all the matrix product is the  matrix product modulo phases. For Clifford operator $C\in \cC_n$, define 
  \begin{align}
     H&:= \cP_n \nonumber,\\
      A &:=  C \{I,Z\}^{\otimes n} C^\dagger \nonumber,\\
      B &:= \cP_{U^\dagger} = U^\dagger \cP_n U \cap \cP_n.
  \end{align}
  Notice that $A$, $B$ and $A\cap B$ are subgroups of $H$, with respect to the matrix-product modulo phases.  Besides, we have
  \begin{align}
      |H| &=4^n, & \nonumber\\
      |A| &=2^n, & \nonumber\\
      |B| &= |\cP_{U^\dagger}|=|\cP_U|= s(U). &\text{ (by Lemma \ref{lem:congs})} 
  \end{align}
Further, we have 
\begin{align}
    |A\cap B| &= |C \{I,Z\}^{\otimes n} C^\dagger \cap U ^\dagger \cP_n U \cap \cP_n| & \nonumber\\
    & =| \{I,Z\}^{\otimes n}  \cap  C^\dagger U^\dagger \cP_n U C \cap \cP_n| \label{eq:a}\\
    &= | \{I,Z\}^{\otimes n}  \cap  C^\dagger U^\dagger \cP_n U C | \label{eq:b}\\
    &\leq s(U\ket{\psi}) \label{eq:ine}. 
\end{align}
The \diff{Eqs.~(\ref{eq:a})~(\ref{eq:b})~(\ref{eq:ine})} use the fact that   $C^\dagger \cP_n C=\cP_n$, $\{I,Z\}^{\otimes n}\in \cP_n$ and Eq.~\eqref{eq:ss} respectively.
By Lemma \ref{lem:pj}, we conclude
\begin{align}
    2^n \cdot s(U) &\leq |A\cap B|\cdot 4^n \label{eq:97},\\
    &\leq s(U\ket{\psi})\cdot 4^n \nonumber,\\
    n-\log_2s(U\ket{\psi}) &\leq 2n -\log_2s(U) \nonumber.
\end{align}
Thus
\begin{align}
    v_s(U\ket{\psi})\leq v(U).\label{eq:D15}
\end{align}

The equality \diff{in (\ref{eq:D15})} holds if and only if equality holds for Eq.~\eqref{eq:ine} and Eq.~\eqref{eq:97}, that is 
\begin{align}
    &\{\bu\in 4^{[n]} \,|\, \langle 0^{\otimes n}C^{\dagger} U^\dagger \sigma_\bu U C|0\rangle^{\otimes n}=\pm 1 \}\nonumber\\
    &=\{\bu\in 4^{[n]}| C^\dagger U^\dagger \sigma_\bu U C\in \pm \{I,Z\}^{\otimes n}\},
\end{align}
and $A\times B=H$.
\end{proof}

\subsection{Detailed proof of Lemma ~\ref{lem:diagplus}}
 \begin{proof}[Proof of Lemma~\ref{lem:diagplus}]
 We prove the inclusion relationships for both sides.
 
 \noindent$\leftarrow$ Suppose $\bu\in 4^{[n]}$ satisfies that $U^\dagger \sigma_\bu U\in \pm \{I,X\}^{\otimes n}$. Since
 \begin{align}
     H^{\otimes n}\{I,X\}^{\otimes n} H^{\otimes n} =\{I,Z\}^{\otimes n}.
 \end{align}
We have
 \begin{align}
     \langle + |^{\otimes n} U^\dagger \sigma_\bu U |+\rangle^{\otimes n} %&= \pm \langle 0^{\otimes n} |\{I,Z\}^{\otimes n} |0^{\otimes n} \rangle\\
     &= \pm 1.
 \end{align}
 
 \noindent$\rightarrow$ Since $U$ is a diagonal unitary, we can write
 \begin{align}
     &U=diag(\cdots, \lambda_\bx,\cdots), \bx\in\{0,1\}^n, |\lambda_\bx|=1.
 \end{align}
 Since
 \begin{align}
    U \ket{+}^{\otimes n}=\frac{1}{\sqrt{2^n}}\sum_{\bx \in \{0,1\}^{n}} \lambda_{\bx}\ket{\bx},
 \end{align}
  then
 \begin{align}
     \langle +|^{\otimes n} U^\dagger \sigma_\bu U\ket{+}^{\otimes n}&=\frac{1}{2^n}\sum_{\bx,\by\in\{0,1\}^n} \lambda_\bx^\dagger \lambda_\by \langle \bx \sigma_\bu \by\rangle.\label{eq:122}.
 \end{align}
 Note that  $|\langle \bx \sigma_\bu \by\rangle|\in\{0,1\}$ and Eq.~\eqref{eq:122} has in total $2^n\times 2^n$ terms. However,  for fixed $\bu$, for any $ \bx\in \{0,1\}^n$, there exists  exactly one $\by_{\bx}\in \{0,1\}^n$ such that $|\langle \bx \sigma_\bu \by_{\bx}\rangle|=1$,
 and other terms are all $0$, that is 
 \begin{align}
     &|\langle \bx \sigma_\bu \by_{\bx}\rangle|=1 \nonumber,\\
     &|\langle\bx \sigma_{\bu}\by\rangle|=0, \forall  \by\neq \by_{\bx}.
 \end{align}
 Thus if $\langle +|^{\otimes n}U^\dagger \sigma_\bu U\ket{+}^{\otimes n}= 1$ (The case for $-1$ is similar), we must have all the nonzero term $\lambda_\bx^\dagger \lambda_{\by_{\bx}} \langle \bx \sigma_\bu \by_{\bx}\rangle=1$. Together with the fact that   $\langle \bx \sigma_\bu \by\rangle=0, \forall \by\neq \by_\bx$. Thus we have
 \begin{align}\label{eq:abseq}
     \lambda_\bx^\dagger \lambda_\by \langle \bx \sigma_\bu \by\rangle=|\lambda_\bx^\dagger \lambda_\by \langle \bx \sigma_\bu \by\rangle|,\quad \forall \bx,\by\in\{0,1\}^{\otimes n}.
 \end{align}
 Note that for any operator $(\cdot)$, the following  holds
 \begin{align}\label{eq:opedec}
     (\cdot)=\sum_{\bx,\by\in\{0,1\}^n} \langle\bx(\cdot) \by\rangle |\bx\rangle \langle \by|. 
     \end{align}
Thus 
\begin{align}
    U^\dagger \sigma_\bu U &=\sum_{\bx,\by\in\{0,1\}^n} \langle\bx( U^\dagger \sigma_\bu U) \by\rangle |\bx\rangle \langle \by| \nonumber\\ 
    &= \sum_{\bx,\by\in\{0,1\}^n} \lambda_\bx^\dagger \lambda_\by  \langle \bx \sigma_\bu \by\rangle |\bx\rangle \langle \by| \nonumber\\ 
    & =\sum_{\bx,\by\in\{0,1\}^n} |\lambda_\bx^\dagger \lambda_\by  \langle \bx \sigma_\bu \by\rangle|\cdot |\bx\rangle \langle \by| \nonumber\\
    & =\sum_{\bx,\by\in\{0,1\}^n} |  \langle \bx \sigma_\bu \by\rangle| \cdot |\bx\rangle \langle \by| \nonumber \\
    &\in\{I,X\}^{\otimes n}.
    \end{align}
The above four equalities and the final inclusion are obtained by Eq.~(\ref{eq:opedec}), the fact that $U\ket{\by}=\lambda_\by \ket{\by}$,  Eq.~(\ref{eq:abseq}), the fact that $|\lambda_\bx|=|\lambda_\by|=1$ and Lemma \ref{lem:pos} respectively.
\end{proof}

\subsection{Detailed proof of Theorem \ref{theo:2n}}
\begin{proof}[of Theorem \ref{theo:2n}]
We prove this theorem by directly calculating the unitary Pauli functions. 
Recall that for any $\bx\in\{0,1\}^n$,
$\odot(\bx):=x_1x_2...x_n$, and $1^n$ is the $n$-bit string of all ones. If we ignore  a global phase $\pm 1$, we may write 
 any $\sigma_\bu,\sigma_\bv\in \cP_n$ as
\begin{align}
    \sigma_\bu =X^{\bs}Z^{\bt},
    \sigma_\bv = X^{\bp}Z^{\bq}.
\end{align}
 where  $\bs,\bt,\bp,\bq\in\{0,1\}^n$. Notice that ignoring the global phase $\pm1$ does not influence the number of $\pm1$s in the Pauli functions.
Write
\begin{align}
    U=C^{n-1}Z\,H^{\otimes n}\, C^{n-1}Z.
\end{align}
\\
\\
One can verify that for any $\bx\in\{0,1\}^n$,
\begin{align}
    &\langle \bx| \sigma_\bu U \sigma_\bv U^\dagger |\bx\rangle \nonumber\\
    &=\frac{1}{2^n} (-1)^{\bs\cdot \bp+\bs\cdot \bt}(-1)^{\bx\cdot(\bp+\bt)}(-1)^{\odot(\bx)+\odot(\bs+\bx)}\nonumber\\
    &\,\quad\times\!\!\!\!\sum_{\by\in\{0,1\}^n}(-1)^{(\bq+\bs)\cdot \by}(-1)^{\odot(\by)+\odot(\by+\bp)}.
\end{align}
Define 
\begin{align}
    f(\bq,\bs,\bp)=\sum_{\by\in\{0,1\}^n}(-1)^{(\bq+\bs)\cdot \by}(-1)^{\odot(\by)+\odot(\by+\bp)}.
\end{align}
We can see that
\begin{align}\label{eq:2n}
    \tr(\sigma_\bu U \sigma_\bv U^\dagger) &= \sum_{\bx\in\{0,1\}^n}  \langle \bx| \sigma_\bu U \sigma_\bv U^\dagger |\bx\rangle \\
    & = \frac{1}{2^n} (-1)^{\bs\cdot \bp+\bs\cdot \bt} f(\bp,\bt,\bs)f(\bq,\bs,\bp).
\end{align}
To further simplify the equations, similar as in \cite{beverland2020lower} we notice that 
\begin{align}
    &  (-1)^{\odot(\by)+\odot(\by+\bp)}=\left\{
    \begin{aligned}
    &  1 \text{\quad$\forall \by$, if $\bp=0$};\\
    &  1  \text{\quad if $\bp\neq 0$;\,  $\by\neq 1^n$ and $\by\neq \bp\oplus 1^n$};\\
    & -1 \text{\quad if $\bp\neq 0$;\, $\by=1^n$ or $\bp+1^n$. }
    \end{aligned}
    \right.
    \end{align}
Combine the fact that
\begin{align}
    \sum_{\by\in\{0,1\}^n} (-1)^{\bz\cdot \by}=\left\{
    \begin{aligned}
    &2^n, \text{ if $\bz=0$,}\\
    &0, \text{ else.}
    \end{aligned}
    \right.
\end{align}
We can see that
\begin{align}\label{eq:fpsq}
    f(\bq,\bs,\bp)=\left\{
    \begin{aligned}
    & \bullet \quad 2^n, \text{ if $\bp=0,\bq=\bs$};\\
    & \bullet \quad 0,  \text{ if $\bp=0,\bq\neq\bs$};\\
    & \bullet \quad 2^n-4,  \text{ if $\bp\neq 0,\bq=\bs$};\\
    &\bullet \quad -2(-1)^{(\bq+\bs)\cdot 1^n}-2(-1)^{(\bq+\bs)\cdot(1^n+\bp)},\\  
    &\quad\quad \text{ if $\bp\neq 0,\bq\neq \bs$.}\\
    \end{aligned}
    \right.
\end{align}
Combine with Eq.~(\ref{eq:2n}) we know that  the unitary Pauli function $\tr(\sigma_\bu U \sigma_\bv U^\dagger)/2^n=\pm 1$ if and only if both $f(\bp,\bt,\bs)=2^n, f(\bq,\bs,\bp)=2^n$, thus by Eq.~(\ref{eq:fpsq}) we conclude this implies 
\begin{align}
    \bp=\bt=\bq=\bs=0.
\end{align}
Thus we conclude that 
\begin{align}
    &s(U)=1 \nonumber,\\
    &v(U)=2n-\log_2s(U)=2n.
\end{align}
\end{proof}

\subsection{Detailed proof of Theorem~\ref{thm:no_more}}

\begin{proof}[Proof of Theorem~\ref{thm:no_more}]
    Since
	for any $d\geq d'$,  we have 
	\begin{align}
		&\max_{C\in \cC_{d+n}} v_s\left( (I_{2^d}\otimes U) C\ket{0}^{\otimes d+n}\right)\nonumber\\
		&\geq \max_{C'\in \cC_{d'+n}} v_s\left( (I_{2^{d'}}\otimes U) C'\ket{0}^{\otimes d'+n}\right). \\
		&\quad\quad \text{ (Appendix \ref{sec:non_incre})}\nonumber
    \end{align}
    It suffices  to give proofs for sufficiently large $d$, namely  $d\geq n$. For any $d\geq n$, by Theorem \ref{lem:compa} we know that
    \begin{align}
    	v(I_d\otimes U)\geq \max_{C\in \cC_{d+n}} v_s\left( (I_{2^d}\otimes U) C\ket{0}^{\otimes d+n}\right)\label{eq:177}.
    \end{align}
  Besides notice that since the unitary-stabilizer nullity satisfies the additivity under tensor product and faithfulness, we have
  \begin{align}
  	v(I_d\otimes U) &= v(I_d) + v(U), &\text{ (Theorem \ref{thm:addten})} \nonumber\\
  	&= 0+v(U). &\text{ (Theorem \ref{prop:faith})}
  \end{align}
  Thus, we have
  \begin{align}
  v(U)\geq \max_{C\in \cC_{d+n}} v_s\left( (I_{2^d}\otimes U) C\ket{0}^{\otimes d+n}\right).\label{eq:greater}
  \end{align}
  Moreover, by Theorem \ref{thm:com_aux} we have
\begin{align}
    v(U) &=v_s\left( I_{2^n}\otimes U \ket{\phi} \right) \nonumber\\
    &\leq \max_{C\in \cC_{2n}} v_s\left( (I_{2^n}\otimes U) C\ket{0}^{\otimes 2n}\right) \nonumber\\
    &\leq \max_{C\in \cC_{d+n}} v_s\left( (I_{2^d}\otimes U) C\ket{0}^{\otimes d+n}\right)\label{eq:smaller}.
  \end{align}
  The last inequality is obtained by noticing $d\geq n$ and using Eq.~(\ref{eq:177}).
  \diff{Combine the above inequalities (\ref{eq:greater}) and (\ref{eq:smaller})}
   we conclude for $d\geq n$,
  \begin{align}
  v(U)= \max_{C\in \cC_{d+n}} v_s\left( (I_{2^d}\otimes U) C\ket{0}^{\otimes d+n}\right).
  \end{align}
  Combine with Eq.~(\ref{eq:177}), we conclude 
  \begin{align}
  v(U) =\max_{d\in \mathbb{N}^+} \max_{C\in\cC_{d+n}} v_s\left( (I_{2^d}\otimes U)\, C\ket{0}^{\otimes (d+n)} \right).
  \end{align}
\end{proof}

%%%%%%%%%%%%%%%
\twocolumngrid
%  \bibliography{smallbib2}

%%%%%%%%%%%%%%%%%%%%%%%%%%%%%%%

 \end{document}